\documentclass[preprint,12pt]{elsarticle}

\usepackage{amssymb}
\usepackage{mathrsfs}

\usepackage{amssymb}
\usepackage{amsmath}
\usepackage{amsthm}

\newcommand{\defeq}{\stackrel{\mbox{\scriptsize{\normalfont\rmfamily def}}}{=}}

\newcommand{\Ch}{\mathrm{Ch}}
\newcommand{\Suc}{\mathrm{Suc}}

\newcommand{\Sub}{\mathrm{Sub}}
\newcommand{\Anc}{\mathrm{Anc}}

\newcommand{\AEEq}{\stackrel{\rm a.e.}{=}}
\newcommand{\vecdiff}[1]{\mathscr{D}_{#1}}

\newcommand{\LPpdf}[5]{#1_{#2}(#3|^{#4}_{#5})}
\newcommand{\LPpdfx}[6]{#1_{#2}(#3|^{#4}_{#5};#6)}
\newcommand{\internal}[1]{#1.I}
\newcommand{\vertex}[1]{#1.V}
\newcommand{\edge}[1]{#1.E}
\newcommand{\Src}[1]{#1.S}
\newcommand{\src}[1]{#1.s}
\newcommand{\Term}[1]{#1.T}
\newcommand{\term}[1]{#1.t}
\newcommand{\JS}{S'}
\newcommand{\jSrc}[1]{#1.\JS}
\newcommand{\JT}{T'}
\newcommand{\jTerm}[1]{#1.\JT}
\newcommand{\removeOE}[2]{#1_{(#2)}} 
\usepackage{bm}
\renewcommand{\Vec}[1]{\bm{#1}}
\newtheorem{theorem}{Theorem}
\newtheorem{lemma}{Lemma}
\newtheorem{definition}{Definition}
\newtheorem{proposition}{Proposition}

\newtheorem{observation}{Observation}
\newtheorem{algorithm}{Algorithm}
\journal{arXiv}

\begin{document}

\begin{frontmatter}



  \title{The Distribution Function of the Longest Path Length in
    Constant Treewidth DAGs with Random Edge Length\tnoteref{t1}}

\tnotetext[t1]{The preliminary results of this paper appeared in \cite{AOSY2009,Ando2017}.}
  
\author{Ei Ando}
\address{Senshu University, 2-1-1, Higashi-Mita, Tama-Ku, Kawasaki, 214-8580, Japan. }
\ead{ando.ei@isc.senshu-u.ac.jp}

\begin{abstract}
  This paper is about the length $X_{\rm MAX}$ of the longest path
  in directed acyclic graph (DAG) $G=(V,E)$ with random edge lengths,
  where $|V|=n$ and $|E|=m$.
  When the edge lengths are mutually independent and
  uniformly distributed, the problem of computing
  the distribution function $\Pr[X_{\rm MAX}\le x]$ is known to be
  $\#$P-hard even in case $G$ is a directed path.
  In this case, $\Pr[X_{\rm MAX}\le x]$ is equal to the volume of
  the knapsack polytope, an $m$-dimensional unit hypercube truncated
  by a halfspace.
  In this paper, 
  we show that there is
  a {\em deterministic} fully polynomial time approximation
  scheme (FPTAS) for computing $\Pr[X_{\rm MAX}\le x]$ in case
  the treewidth of $G$ is at most a constant $k$.
  The running time of our algorithm is
  $O(k^2 n(\frac{16(k+1)mn^2}{\epsilon})^{4k^2+6k+2})$
  to achieve a
  multiplicative approximation ratio $1+\epsilon$.
  Before our FPTAS, we present a fundamental formula
  that represents $\Pr[X_{\rm MAX}\le x]$ by at most
  $n-1$ repetitions of definite
  integrals.
  Moreover,
  in case the edge lengths follow the mutually independent standard
  exponential distribution, we show 
  a $((4k+2)mn)^{O(k)}$ time
  exact algorithm.
  For random edge lengths satisfying certain conditions,
  we also show
  that computing $\Pr[X_{\rm MAX}\le x]$ is
  fixed parameter tractable if we choose treewidth $k$,
  the additive error $\epsilon'$, and $x$ as the parameters.
\end{abstract}




\end{frontmatter}


\section{Introduction}
We consider the longest path length $X_{\rm MAX}$ of
a directed acyclic graph (DAG) $G$, where
the edge lengths of $G$ are mutually independent random variables.
It is well known that
the longest path problem in DAGs with static edge lengths can
be solved in linear time in the graph size \cite{I2A}.
In this paper, however, we introduce random edge lengths.
There are at least two reasons.
Firstly, there are a lot of uncertain quantities in the industry.
Therefore,
it is often meaningful to consider the behavior of the uncertainty
using the random variables.
Secondly, from the
viewpoint of computational complexity,
we expect that the probability of a complex event involved with
a graph may lead to the difference of computational
performance between
a randomized algorithm and a deterministic algorithm.
Though it is widely believed that BPP=P~~\cite{AB2009},
randomized computation still seems to have some advantages in
approximation~~\cite{Elekes1986}.
Thus, we are interested in the difference
in the computational performance 
caused by randomness, especially
in computing the probability, or the high dimensional volume.
In this paper, we show a fundamental formula for the distribution
function of the longest path length in $G$ with random edge lengths.
Then, we describe
three deterministic algorithms,
including a deterministic fully polynomial time approximation scheme
(FPTAS) for a $\#$P-hard problem.

As far as the author is aware, the longest path problem in DAGs with
random edge lengths is proposed by Malcolm et al.~\cite{MRCF1959} in
the field of operations research.
For some special classes of DAGs such as series-parallel graphs,
some deterministic and exact polynomial time algorithms are known as
folklore~(e.g., \cite{CJ2016}).
However, this approach does not scale to the general DAG.
For general DAGs, some deterministic exponential time algorithms
are proposed at first~(e.g., \cite{Martin1965}).

The longest path problem in DAGs with random edge lengths
has been well studied in the field of VLSI design
(e.g., \cite{BCSS2008,CJ2016}).
The time difference (signal delay) 
between the input and the output of each logical circuit product 
may vary among mass-produced semiconductor chips
even though they are from the same line of the same design.
The signal delay fluctuates because the signal delay of each
logical gate fluctuates, which is inevitable to some extent. 
Therefore, before they start costly mass-production,
the VLSI makers would like to know
whether or not a sufficient number of their new chips
will perform as expected. 
To estimate the signal delay of a logical circuit, we consider 
the longest path length in a DAG by considering each of the gates and lines
as an edge and each fluctuating signal delay as a random edge length.
Then, the signal delay in the entire circuit is
the longest path length in the DAG.
To estimate how many products perform well, 
we would like to know the probability (the distribution function)
that the length of the longest path is at most a given value $x$.

\subsection{Formal Description of the Problem}
We consider a DAG $G=(V,E)$ with vertex set
$V=\{1,\dots,n\}$ and edge set $E\subseteq V\times V$ where $|E|=m$.
When an edge $e$ goes from $u\in V$ to $v\in V$, we write $e=uv=(u,v)$.
For any graph $G$, we write $\vertex{G}$ and $\edge{G}$ to mean
the vertex set and the edge set of $G$, respectively.
We assume that $G$ has no isolated vertex.
We also
assume that the vertex set $V=\{1,\dots,n\}$ are topologically ordered.
\begin{definition}
For a DAG $G$, we define a {\em source} of $G$ as
a vertex that has no incoming edge.
Also, a {\em terminal} of $G$ is a vertex that has no outgoing edge.
The set of all sources (resp. terminals) in $G$ is denoted by
$\Src{G}=\{1,\dots,|\Src{G}|\}$
(resp. $\Term{G}=\{n-|\Term{G}|+1,\dots,n\}$).
When a vertex $v\in \vertex{G}$ is neither a source nor a terminal,
$v$ is an {\em internal vertex}, and
$\internal{G}=\vertex{G}\setminus (\Src{G}\cup \Term{G})$
is the set of all internal vertices.
A (directed) $s-t$ path $\pi$ in $G$ is a subgraph of $G$ where
$\vertex{\pi}=\{u_1,u_2,\dots,u_p\}\subseteq \vertex{G}~~(p\ge 2)$,
$\edge{\pi}=\{u_iu_{i+1}|i=1,\dots,p-1\}$,
$u_1\in \Src{G}$ and $u_p\in \Term{G}$.
We write $\Pi_G$ to mean
the set of all paths from the sources to the terminals in $G$.
\end{definition}

Then, we consider $m$ mutually independent random variables.
Let $Y_{uv}$ for each $e=uv\in E$ be a mutually independent random variable.
We compute the probability that the longest path length 
$X_{\rm MAX}=\max_{\pi\in \Pi_G}\left\{\sum_{uv\in \edge{\pi}}Y_{uv}\right\}$ 
is at most $x\in \mathbb{R}_{\ge 0}$.
Note that computing $\Pr[X_{\rm MAX}\le x]$ and finding out
a path that is the ``longest'' in some sense are different problems
in the case of random edge lengths.
It may seem counter-intuitive because these two problems are
almost equivalent if edge lengths are static values.
However, any $s-t$ path in $G$ can be the longest with
a certain positive probability
when we assume some natural distributions
for the random edge lengths 
like the uniform distribution over $[0,1]$.
Throughout this paper, we focus on the problem of estimating
the probability distribution function $\Pr[X_{\rm MAX}\le x]$.
We do not consider the problem of finding out any path.

In case we consider that the length of each edge $uv\in E$
is uniformly distributed,
the input of the problem consists of
DAG $G$, $x\in\mathbb{R}$, and a vector $\Vec{a}\in \mathbb{Z}_{>0}^m$.
Here, for each edge $uv\in E$, $\Vec{a}$ has a component
$a_{uv}\in\mathbb{Z}_{>0}$.
Let a random vector $\Vec{X}$ be
uniformly distributed over $[0,1]^m$.
For each edge $uv\in E$,
random variable $Y_{uv}=a_{uv}X_{uv}$ 
is the random edge length with
its distribution function
$F_{uv}(x)=\Pr[a_{uv}X_{uv}\le x]$.
In this case, the distribution function $\Pr[X_{\rm MAX}\le x]$
of the longest path length 
is equal to the volume of a polytope 
\begin{align*}
  K_G(\Vec{a},x)\defeq \left\{\Vec{x}\in [0,1]^m \left| \bigwedge_{\pi\in \Pi_G}\sum_{uv\in \edge{\pi}} a_{uv}x_{uv}\le x\right.\right\},
\end{align*}
where each component $x_{uv}$ of $\Vec{x}\in \mathbb{R}^{m}$ is
related to an edge $uv\in E$.
If $G$ is a 
directed path, $K_G(\Vec{a}, x)$ is a $0-1$ knapsack polytope.
Computing the volume of $K_G(\Vec{a},x)$ is $\#$P-hard 
even if $G$ is a directed path (see~\cite{DF1988}) 
\footnote{
Intuitively, the
breakpoints of the function
$F(x)={\rm Vol}(K_G(\Vec{a},x))$ increases 
exponentially with respect to $n$. For example,
consider the case where each component $a_{i,i+1}$ 
of $\Vec{a}$ is $a_{i,i+1}=2^i$ for edge $(i,i+1)~~i=1,\dots,n-1$.
}.

The different assumptions of the edge length
distributions give us completely different problems.
For example, 
if the edge lengths are all normally distributed or
all exponentially distributed,
and $G$ is a directed path,
we can efficiently compute $\Pr[X_{\rm MAX}\le x]$.
If the edge lengths are discrete random variables,
Hagstrom~\cite{Hagstrom1988} proved that the problem is $\#$P-hard
in DAGs where any $s-t$ path has at most three edges.

\subsection{Results about Computing High Dimensional Volume}
Consider uniformly distributed edge lengths.
Then, the problem of computing the value of the distribution function of
the longest path length is equivalent to computing the volume of a polytope
$K_G(\Vec{a},x)$.
Here, we briefly introduce some results about
randomized approximation algorithms, hardness results, and
deterministic approximation algorithms.

There are efficient randomized algorithms
for estimating $n$-dimensional convex volumes.
Dyer et al.~\cite{DFK1991}
showed the first FPRAS (fully polynomial time randomized approximation scheme)
that finishes in $O^\ast(n^{23})$ time for the volume of 
the general $n$-dimensional convex body. Here $O^\ast$ ignores 
the factor of ${\rm poly}(\log n)$ and $1/\epsilon$ factor.
There are faster
FPRASs \cite{LV2006,CV2015}. 
Currently, the fastest FPRAS\cite{CV2015} 
runs in $O^\ast(n^3)$ time for well-rounded convex bodies.

In general, exactly computing the $n$-dimensional volume of a polytope
is hard.
Elekes~\cite{Elekes1986} considered
an $n$-dimensional convex body accessible by membership oracle
and proved that no deterministic polynomial
time algorithm achieves the approximation ratio of $1.999^n$.
The bound is updated by 
B\'{a}r\'{a}ny and ~F\"{u}redi~\cite{BF1987} up to $(c n/\log n)^{n/2}$, 
where $c$ does not depend on $n$.
Dyer and Frieze~\cite{DF1988} proved that computing the volume of 
the $0-1$ knapsack polytope $K$ is $\#$P-hard.

Recently, there have been
some deterministic approximation algorithms for the volume
of the knapsack polytope.
Li and Shi~\cite{LS2014} showed
a fully polynomial time approximation scheme (FPTAS)
for the distribution function of 
the sum of the discrete random variables. 
Their algorithm is applicable to the approximation of
${\rm Vol}(K_G(\Vec{a},x))$ 
if $G$ is a directed path.
Their algorithm is based on
some ideas due to 
\v{S}tefankovi\v{c} et al.~\cite{SVV2012}
(see also \cite{GKM2010},\cite{GKMSVV2011}),
which is a kind of dynamic programming.
Motivated by the deterministic approximation technique
of the above results, Ando and Kijima~\cite{AK2016}
showed another FPTAS
using the approximate convolution integral.
Their algorithm runs in $O(n^3/\epsilon)$ time.
They extended the FPTAS to the problem of computing the volume of 
the multiple constraint knapsack polytope. 
Given $m\times n$ matrix $A\in \mathbb{Z}_{\ge 0}^{mn}$ and
a vector $\Vec{b}\in \mathbb{Z}_{\ge 0}^{m}$,
the multiple constraint knapsack polytope 
is $K_m(A,\Vec{b})\defeq \{\Vec{x}\in [0,1]^{n}| A\Vec{x}\le \Vec{b}\}$.
Their algorithm finishes in $O((n^2\epsilon^{-1})^{m+1}nm \log m)$ time. 
Thus, there is an FPTAS for ${\rm Vol}(K_m(A,\Vec{b}))$
if the number $m$ of constraints is bounded by a constant.

\subsection{Treewidth and Related Results}
Robertson and Seymour~\cite{RS1986} first
defined the notion of the treewidth.
For the case where the treewidth of undirected graph $G$
is at most a constant $k$,
Bodlaender's linear time algorithm~\cite{Bodlaender1996} 
finds the tree decomposition of $G$ in linear time.
Many NP-hard problems and $\#$P-hard problems on graphs
are solvable in polynomial time when the treewidth $k$ is bounded.
See \cite{DoFe1999} for classic results.
Courcelle and Engelfriet~\cite{CourcelleEngelfriet}
proved the existence of a linear time algorithm for
any graph optimization problem that can be described by
Monadic Second Order Logic (MSOL) if the treewidth of $G$
is constant.
Johnson et al.~\cite{JRST2001} defined the directed treewidth. 
They proved that various NP-hard problems,
including computing the directed treewidth,
can be solved in polynomial time on the directed graphs 
with at most constant directed treewidth.

We may expect that computing 
$\Pr[X_{\rm MAX}\le x]$ can also be
solved when the treewidth of $G$ in some sense is bounded by
a constant.
Previously, however, it has not been clear
how to do it, 
probably because the problem
lies at the junction between 
the continuous analysis and the combinatorial optimization.
Our goal is to develop the techniques for this kind of problem.

In \cite{AOSY2009}, Ando et al. considered the problem of computing
$\Pr[X_{\rm MAX}\le x]$ in DAGs with random edge lengths.
They showed a formula that represents $\Pr[X_{\rm MAX}\le x]$
with $n-1$ repetitions of definite integrals and
presented some results in case the antichain size of
the incidence graph $G$
is at most a constant.
In that setting, they proved that
the exact computation of
$\Pr[X_{\rm MAX}\le x]$ finishes in polynomial time
if the edge lengths are exponentially distributed with parameter $1$
(standard exponential distribution).
They also considered the Taylor approximation of $\Pr[X_{\rm MAX}\le x]$,
assuming the random edge lengths follow some
implicit distributions satisfying
certain conditions
where the edge lengths
may not be exponentially distributed.

\subsection{Contribution}
The first contribution of this paper is 
a formula that represents $\Pr[X_{\rm MAX}\le x]$ by
$n-|\Term{G}|$ repetitions of definite integrals.
Then, the formula leads us to
an FPTAS for computing the value
of $\Pr[X_{\rm MAX}\le x]={\rm Vol}(K_G(\Vec{a},x))$
for DAG $G$, whose underlying undirected graph of $G$
has treewidth bounded by a constant $k$
and the edge lengths follow the uniform distribution.
Extending \cite{Ando2017},
which considers an FPTAS for $G$ with constant pathwidth,
we will prove the following theorem
by showing how to deal with multiple subtrees in the tree decomposition.
\begin{theorem}
  Suppose that the treewidth of the underlying undirected graph of
  DAG $G$ is at most a constant $k$.
  There is an algorithm that approximates ${\rm Vol}(K_G(\Vec{a},x))$
  in $O\left(k^2 n\left(\frac{16(k+1)mn^2}{\epsilon}\right)^{4k^2+6k+2}\right)$ time
  satisfying $1\le  V'/{\rm Vol}(K_G(\Vec{a},x))\le 1+\epsilon$, where 
  $V'$ is the output of the algorithm.
\end{theorem}

Next, we show more examples where
the edge lengths follow other distributions.
If the edge lengths are mutually independent and
follow the standard exponential distribution,
we present how to compute the exact value of
$\Pr[X_{\rm MAX}\le x]$. 
Previously, 
no polynomial time algorithm has been known
even in case the treewidth of $G$ is bounded by a constant. 
We next consider more general continuous distribution for
the random edge lengths. In this setting, we obtain the value of
$\Pr[X_{uv}\le x]$ and its derivative of arbitrary order for
$uv\in E$ by an oracle.
If the edge length distributions satisfy certain
conditions, we can use the Taylor approximation
for $\Pr[X_{uv}\le x]$.
We show that the problem of computing
$\Pr[X_{\rm MAX}\le x]$ is fixed parameter tractable
if we choose treewidth $k$, additive error $\epsilon'$, and $x$
as the parameter. 
Our technique applies to
the standard exponentially distributed edge lengths and
the Taylor approximation,
which extends the result of \cite{AOSY2009}
to the graphs with constant treewidth.
Note that the maximum antichain size of the incidence graph can be
$\Omega(n)$ even if $G$ has constant treewidth.

\subsection{Organization}
Section 2 shows a fundamental formula representing
$\Pr[X_{\rm MAX}\le x]$ by $n-|\Term{G}|$ repetitions of definite integrals.
In Section 3, we describe our approximation algorithm for
the volume of $K_G(\Vec{a},x)$. Then,
we prove the approximation ratio and the running time so that
our algorithm is an FPTAS. 
Section 4 shows the cases where
the edge lengths follow other distributions. 
There, we present an exact
$((4k+2)mn)^{O(k)}$-time algorithm for the
case where edge lengths follow the standard exponential distribution.
Finally, we show another result assuming
abstract distributions for the
edge lengths satisfying certain conditions.
We prove that computing
$\Pr[X_{\rm MAX}\le x]$ is a fixed parameter tractable problem
if $k$, $\epsilon'$, and $x$ are the parameters.

\section{Integrals that Give $\Pr[X_{\rm MAX}\le x]$}
In this section, we prove a fundamental formula for $\Pr[X_{\rm MAX}\le x]$.
Since the formula includes definite integrals repeated $n-|\Term{G}|$ times,
we deal with $n-|\Term{G}|$ dummy variables of integrals.
We will see that each of these dummy variables is associated with
a vertex of $G$.
Then, we show how to compute $\Pr[X_{\rm MAX}\le x]$ by
using the tree decomposition.
To exploit the advantage of tree decomposition,
we need our notation for vectors. 

\subsection{Dummy Variable Vectors}
Note that our problem lies at the junction between
the continuous analysis and the combinatorial optimization.
We need a set of notations 
for concisely arguing the combinations
of many partial derivatives and integrals.
Here, we define notations for the vectors of dummy variables
using an element or a set as its subscript of a component
or a sub-vector.
Let $W$ be a subset of $V$ or a subset of $E$ where $G=(V,E)$.
We consider a vector $\Vec{z}\in \mathbb{R}^{|W|}$.
Each component of $\Vec{z}$ is specified by an element of $W$.
That is, $\Vec{z}[v]$ is a component of $\Vec{z}$ for $v\in W$.
In the following, we write $\Vec{z}\in \mathbb{R}^W$.
For simplicity, we define
$\Vec{z}[v]=0$ for any $v\not\in W$.

For the conciseness, we pack the dummy variables of
the integrals of the ongoing computation as a vector.
In the following, we use $\Vec{z}$ as an $n$-dimensional vector of 
variables $\Vec{z}=(z_1,\dots,z_n)$,
where each component $z_v$ is associated with $v\in V$ 
(i.e., $\Vec{z}[v]=z_v$).
We write $\Vec{y}=\Vec{z}[W]$ for some $W\subseteq V$.
Here, $\Vec{y}$ is a vector with $|W|$ components where
$\Vec{y}[v]=\Vec{z}[v]$ for $v\in W$.
We use the same notation for constant vectors $\Vec{0}$ and $\Vec{1}$, 
where $\Vec{0}[v]=0$ and $\Vec{1}[v]=1$ for any $v$. 
Consider the definite integrals with $|W|$ dummy variables
of $|W|$-variable function $F(\Vec{z}[W])$.
We write 
\begin{align*}
  \int_{\mathbb{R}^W}F(\Vec{z}[W]) {\rm d}\Vec{z}[W]=
  \int_{\mathbb{R}}\cdots \int_{\mathbb{R} }F(\Vec{z})~{\rm d}\Vec{z}[v_{|W|}]\cdots {\rm d}\Vec{z}[v_1],
\end{align*}
for $W=\{v_1,\dots,v_{|W|}\}$.
By {\em arguments} of $F(\Vec{z}[W])$, we mean
$\Vec{z}[W]\in \mathbb{R}^W$ for
the $|W|$-variable function $F(\Vec{z}[W])$.
Since we also deal with
partial derivatives
with respect to many variables,
we save some space by writing
\begin{align*}
  \vecdiff{W}=F(\Vec{z})\defeq\frac{\partial}{\partial \Vec{z}[W]} F(\Vec{z}).
\end{align*}
Let $f_U(\Vec{z}[W])=\vecdiff{U}F(\Vec{z}[W])~~(U\subseteq W)$.
Consider the inverse operation of $\vecdiff{U}$.
For $z_{v_1},\dots,z_{v_{|U|}}$, clearly we have
\begin{align*}
  F(\Vec{z}[W])&=\int_{-\infty}^{\Vec{z}[v_1]}\cdots\int_{-\infty}^{\Vec{z}[v_{|U|}]}f_U(\Vec{z}[W]){\rm d}\Vec{z}[v_{|U|}]\cdots {\rm d}\Vec{z}[v_1], \hspace*{10mm}\text{or,}\\
  F(\Vec{z}[W])&=\int_{(-\infty,\Vec{z}[U])}\hspace*{-10mm}f_U(\Vec{z}[W]){\rm d}\Vec{z}[U],~~\text{where}~~(-\infty, \Vec{z}[U])\defeq\{\Vec{x}\in \mathbb{R}^{U}|\Vec{x}\le \Vec{z}\}.
\end{align*}

In this paper, we use convolutions extensively.
Let $F_v(x)$ and $f_v(x)$ be the probability distribution
function and the probability density function of
random variable $X_v$ for $v=1,2$.
When $X_1$ and $X_2$ are mutually independent,
the distribution function of $X_1+X_2$ is
$\Pr[X_1+X_2\le x]=\int_{\mathbb{R}}f_1(z)F_2(x-z){\rm d}z$.
The density function is 
$\frac{d}{dx}\Pr[X_1+X_2\le x]=\int_{\mathbb{R}}f_1(z)f_2(x-z){\rm d}z$.
We call these integration operations as {\em convolutions}.
We also use $\max$ of $X_1$ and $X_2$.
It is well known that $\Pr[\max\{X_1,X_2\}\le x]=F_1(x)F_2(x)$.

As the ending of the preliminaries, we here introduce one special function.
\begin{definition}
  Let $H(x)$ be {\em Heaviside step function},
  where $H(x)=0$ for $x<0$ and $H(x)=1$ for $x\ge 0$.
\end{definition}
The step function is, in a sense, the distribution function of
a ``random'' variable that takes a static value $0$ with probability $1$.
We use Dirac's delta function $\delta(x)$ as the density.

\subsection{Repetition of Definite Integrals for the Distribution Function of $X_{\rm MAX}$}
We prove that $\Pr[X_{\rm MAX}\le x]$ can be given by
$n-|\Term{G}|$ repetitions of definite integrals.
We first define the longest path length starting at each vertex 
as a random variable.
The following graph plays an essential role.
\begin{definition}
  For $U\subseteq \vertex{G}$,
  graph $\removeOE{G}{U}$ is obtained from
  $G$ by removing all outgoing edges $uv$ of all $u\in U$
  and removing all isolated vertices.
\end{definition}
Next, by using the topologically last internal vertex $v$,
we transform $\Pr[X_{\rm MAX}\le x]$ into 
a convolution of two functions:
the distribution function of the longest path length in 
$\removeOE{G}{v}$ and
that of the longest path length from $v$ to any terminal of $G$.
By repeating the removal of the outgoing edges from
the topologically last internal vertex
until the remaining graph has no edge,
we have the formula for $\Pr[X_{\rm MAX}\le x]$.
In the resulting form,
we associate a variable with each of the sources and
also with each of the terminals,
which will be important later.

Though one obstacle in
computing $\Pr[X_{\rm MAX}\le x]$ is that 
the path lengths can be dependent on each other
because of the shared edges,
we avoid this by considering 
the longest path length from the topologically
last internal vertex $v$.
Let $Z_v$ be the longest path length
from vertex $v$.
Then, $Z_v$ and the edge lengths 
from any topologically earlier vertex $u$
are mutually independent.
Here, we define $\Suc(v)\defeq\{w\in V|vw\in E\}$.
We define $Z_v$ as follows.
\begin{definition}
  \label{definition:Zu}
  For $u\in V$, we set 
  $Z_u\defeq \max_{v\in \Suc(u)}\{X_{uv}+Z_v\}$,
  where $Z_t$ for any $t\in \Term{G}$ is equal to a given 
  value
  $z_t=0$.
\end{definition}
Though we set $z_t=0$ here, we define $z_t$ as a real variable later,
which works in merging the partial results for two graphs. 
It is well known that $Z_u$ is the longest path length
(e.g., Lemma 24.17 of \cite{I2A}).
\begin{proposition}
  \label{proposition:Zu}
The longest path length from $u$ to any terminal is $Z_u$.
\end{proposition}

We define the conditional 
distribution function and the density function of $Z_v$
assuming $Z_w=z_w$ for $w\in\Suc(v)$.
Later, $z_w$'s work as the dummy variables
of the convolutions. 
\begin{definition}
  Let $F_{uv}(x)=\Pr[X_{uv}\le x]$ for $uv\in E$.
  We define
  \begin{align*}
    \LPpdf{F}{u}{\Vec{z}}{u}{\Suc(u)}&\defeq \Pr[Z_u\le z_u| Z_v=z_v, \forall v\in \Suc(u)] = \hspace*{-2mm}\prod_{v\in \Suc(u)}\hspace*{-2mm}F_{uv}(z_u-z_v),
  \end{align*}
  and $\LPpdf{f}{u}{\Vec{z}}{u}{\Suc(u)}\defeq \frac{\partial}{\partial z_u} \LPpdf{F}{u}{\Vec{z}}{u}{\Suc(u)})$.
\end{definition}
Then, we define
$|\Src{G}\cup \Term{G}|$-variable functions $\Phi$ and $\phi$.
\begin{definition}
  \label{definition:Phi}
  For an $s-t$ path $\pi$ in $G$, let $\src{\pi}$ and $\term{\pi}$ be the source
  and the terminal of $\pi$.
  Remember that $\Pi_G$ is the set of all $s-t$ paths. Then,
  \begin{align*}
    \LPpdf{\Phi}{G}{\Vec{z}}{S}{T}=\Phi(G;\Vec{z}[\Src{G}],\Vec{z}[\Term{G}])&\defeq\!\Pr\!\left[\bigwedge_{\pi \in \Pi_G}\hspace*{-2mm}\sum_{\hspace*{2mm}uv\in \edge{\pi}}\hspace*{-3mm}X_{uv}\!\le\!\Vec{z}[\src{\pi}]\!-\!\Vec{z}[\term{\pi}]\right]\!\!,~\text{and}\\
    \LPpdf{\phi}{G}{\Vec{z}}{S}{T}=\phi(G;\Vec{z}[\Src{G}],\Vec{z}[\Term{G}])&\defeq\vecdiff{\Src{G}}\Phi(G;\Vec{z}[\Src{G}],\Vec{z}[\Term{G}]).
  \end{align*}
\end{definition}
Remember that $\vecdiff{\Src{G}}$
represents partial differentiation
with respect to all $z_s$'s for $s\in \Src{G}$.
Here, $\Vec{z}[\Src{G}]$ specifies the
longest path length from each of $s\in\Src{G}$;
$\Vec{z}[\Term{G}]$ is the value
we subtract from the longest path length ending at 
each of $t\in \Term{G}$.
That is, any path starting at $s\in \Src{G}$ and 
ending at $t\in \Term{G}$
must be shorter than $\Vec{z}[s]$ by $\Vec{z}[t]$,
which works as a kind of glue to the other graphs later.
Let $\Vec{\sigma}$ be a vector satisfying
$\Vec{\sigma}[s]=1$ for $s\in \Src{G}$ and
$\Vec{\sigma}[v]=0$ for $v\not\in\Src{G}$.
Then,
\begin{align*}
  \Pr[X_{\rm MAX}\le x]&=\LPpdf{\Phi}{G}{x\Vec{\sigma}}{S}{T}=\int_{(-\infty,x\Vec{1}[\Src{G}])}\LPpdf{\phi}{G}{x\Vec{\sigma}}{S}{T}{\rm d}\Vec{z}[\Src{G}].
\end{align*}

Theorem \ref{th:multiplesourceterminal} shows
the exact computation of $\LPpdf{\phi}{G}{\Vec{z}}{S}{T}$. 
\begin{theorem}
  \label{th:multiplesourceterminal}
  $\displaystyle\LPpdf{\phi}{G}{\Vec{z}}{S}{T}=\int_{\mathbb{R}^{\internal{G}}}\prod_{u\in \Src{G}\cup\internal{G}}\LPpdf{f}{u}{\Vec{z}}{u}{\Suc(u)} {\rm d}\Vec{z}[\internal{G}]$
\end{theorem}
\begin{proof} 
Remember that $\Vec{z}[s]=z_s$. By Proposition~\ref{proposition:Zu}, we have 
\begin{align*}
  \bigwedge_{\pi \in \Pi_G}\hspace*{-2mm}\sum_{\hspace*{2mm}uv\in \edge{\pi}}\hspace*{-2mm}X_{uv}\!\le\!\Vec{z}[\src{\pi}]-\Vec{z}[\term{\pi}]
  \hspace*{3mm}\Leftrightarrow \bigwedge_{u\in \vertex{G}}\hspace*{-4mm}\bigwedge_{\hspace*{4mm}v\in\Suc(u)} \hspace*{-5mm}X_{uv}+Z_v\!\le\! Z_u \wedge \bigwedge_{s\in\Src{G}}\hspace*{-2mm}Z_s\!\le\! z_s. 
\end{align*}
The condition on the right is due to
Definition~\ref{definition:Zu}.
We use the condition on the right instead of the condition in the probability of Definition~\ref{definition:Phi}. 
That is,
\begin{align*}
  \LPpdf{\Phi}{G}{\Vec{z}}{S}{T}=\Pr\left[\bigwedge_{u\in \vertex{G}}\bigwedge_{v\in \Suc(u)}X_{uv}+Z_v\le Z_u\wedge \bigwedge_{s\in \Src{G}} Z_s\le z_s\right].
\end{align*}


We move on to the proof of the formula.
For any $W\subseteq \internal{G}$, we prove the following formula 
by the induction on $W$.
\begin{align}
  \LPpdf{\Phi}{G}{\Vec{z}}{S}{T}=\int_{\mathbb{R}^W}\LPpdf{\Phi}{\removeOE{G}{W}}{\Vec{z}}{S}{T} \prod_{w\in W}\LPpdf{f}{w}{\Vec{z}}{w}{\Suc(w)}{\rm d}\Vec{z}[W]. \label{form:inductionhypothesis}
\end{align}
That is, we aggregate all possible values of $Z_w$'s for $w\in W$
by introducing integrals with respect to $z_w$'s, where
the density function of $Z_w$ is
$\LPpdf{f}{w}{\Vec{z}}{w}{\Suc(w)}$.

As the base case, we have (\ref{form:inductionhypothesis}) 
for $W=\emptyset$.
In the induction step, we add the vertices $v\in \internal{G}$ to $W$
in the backward topological order.
Let $v\in \internal{G}\setminus W$
where $W$ is the set of
internal vertices that are topologically later than $v$.
Since each vertex $w\in W$ is a terminal or removed as an isolated vertex in
$\removeOE{G}{W}$, we already have the integral with respect to $z_w$
that aggregates all possible values of $Z_w$.
Then, all $z_w$'s for $w\in W$ are static values in the integrand,
so that  
\begin{align*}
  \LPpdf{\Phi}{\removeOE{G}{W}}{\Vec{z}}{S}{T}&=\int_{\mathbb{R}}\LPpdf{\Phi}{\removeOE{G}{W\cup \{v\}}}{\Vec{z}}{S}{T}\LPpdf{f}{v}{\Vec{z}}{v}{\Suc(v)}{\rm d}\Vec{z}_v, \\
  \LPpdf{\Phi}{G}{\Vec{z}}{S}{T}&=\int_{\mathbb{R}^{W}}\int_{\mathbb{R}}\LPpdf{\Phi}{\removeOE{G}{W\cup \{v\}}}{\Vec{z}}{S}{T} \LPpdf{f}{v}{\Vec{z}}{v}{\Suc(v)}{\rm d}\Vec{z}[v]\prod_{w\in W}\LPpdf{f}{w}{\Vec{z}}{w}{\Suc(w)}{\rm d}\Vec{z}[W],
\end{align*}
implying that the induction hypothesis (\ref{form:inductionhypothesis})
holds for $W\cup \{v\}$.

Therefore, we repeat the transformation
for all $v\in \internal{G}$ in backward topological order, which leads to
\begin{align*}
  \LPpdf{\Phi}{G}{\Vec{z}}{S}{T}=\int_{\mathbb{R}^{\internal{G}}}\LPpdf{\Phi}{\removeOE{G}{\internal{G}}}{\Vec{z}}{S}{T}\prod_{v\in \internal{G}}\LPpdf{f}{v}{\Vec{z}}{v}{\Suc(v)}{\rm d}\Vec{z}[\internal{G}].
\end{align*}
Since $\LPpdf{\Phi}{\removeOE{G}{\internal{G}}}{\Vec{z}}{S}{T}=\prod_{s\in \Src{G}}\LPpdf{F}{s}{\Vec{z}}{s}{\Suc(s)}$, we obtain the theorem.
\end{proof}

\subsection{Tree Decomposition}
The followings are the definitions of the treewidth and related terms.
For any set $W$, $\binom{W}{k}$ is the family of subsets $U$ of $W$
with cardinality $|U|=k$.
\begin{definition}
  \label{def:treedecomposition}
  A {\em tree decomposition} of $G$ is
  a pair ${\cal T}(G)=({\cal B}, {\cal A})$,
  where ${\cal B}=\{B_0,B_1,\dots,B_{b-1}\}$ is a family of subsets of $V$,
  and
  ${\cal A}\subseteq \binom{\cal B}{2}$ is an edge set so that ${\cal T}(G)$
  is a tree satisfying the following three conditions. 
\begin{enumerate}[{\hspace*{2mm}TD}1{.\hspace*{1mm}}]
  \item $\bigcup_{B\in {\cal B}}B = V$;
  \item $uv\in E\Rightarrow\exists B\in {\cal B}$ s.t. $\{u,v\}\subseteq B$;
  \item for all $B_h,B_i,B_j\in {\cal B}$, if $B_i$ is on a path from $B_h$ to $B_j$ in ${\cal T}(G)$, then $B_h\cap B_j\subseteq B_i$.
\end{enumerate}
We call $B\in {\cal B}$ a {\em bag}.
The {\em width} of 
${\cal T}(G)$ is $\max_{B\in {\cal B}}|B|-1$.
The {\em treewidth} of $G$ is the minimum of the width of
all possible ${\cal T}(G)$'s. 
\end{definition}
Fig.\ref{fig:FoldedPath} shows an example of a DAG and its tree decomposition.
\begin{figure}
  \begin{center}
    \includegraphics[width=70mm,clip]{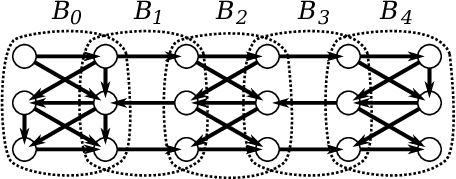}
  \end{center}
  \vspace*{-5mm}
  \caption{An example of a DAG and its tree decomposition. Bags $B_1,B_2$, and $B_3$ are visited three times by a path. By the ancestor-first rule, the edges between the vertices in $B_0\cap B_1$ belong to bag-subgraph $G_0$. The bag-subgraph is defined in Section \ref{section:bagsubgraph}.}
  \label{fig:FoldedPath}
\end{figure}

In the tree decomposition, we call $B_0\in {\cal B}$
the {\em root} of ${\cal T}(G)$.
Let $\Ch(i)$ be the index set of the children of 
bag $B_i\in {\cal B}$ in ${\cal T}(G)$.
Let $\Sub(i)$ be the index set of bags in the subtree of ${\cal T}(G)$
rooted at $B_i$.

We may expect 
an efficient algorithm
if $G$ has small treewidth.
Small treewidth allows us
to process a bag in a brute-force manner. 
Therefore, if we can efficiently process the leaves first
and then merge the subtrees,
we have an efficient algorithm.
We use the following {\em binary tree decomposition}
for bounding the running time of our algorithm.
\begin{proposition}[Lemma 13.1.3 of \cite{Kloks1994}]
  \label{proposition:binary}
  Given a tree decomposition ${\cal T}(G)$ of $G$ of width at most $k$, 
  it is possible to transform ${\cal T}(G)$ in time $O(n)$ into a
  binary tree decomposition ${\cal T}'(G)$ of $G$ of width $k$ and
  with at most $4n$ bags, where 
  we have $|\Ch(i)|\le 2$ for any bag $B_i$ in ${\cal T}'(G)$.
\end{proposition}
Throughout this paper, we assume that
$G$ has $n$ vertices and that 
the treewidth of $G$ is bounded by a constant $k$.
We assume that the number $b$ of bags is at most $4n$,
and that the number of the children of a bag is at most two.
In the following, we set $\Ch(i)=\{\ell,r\}$ when
$|\Ch(i)|=2$.

\subsection{Bag-Subgraph}
\label{section:bagsubgraph}
We here define a {\em bag-subgraph} $G_i$ for each bag $B_i\in {\cal B}$.
Notice that $G_i$ is not
the induced subgraph with vertex set $B_i$.
See 
Fig.~\ref{fig:FoldedPath},
where we have edges between the vertices in $B_0\cap B_1$.
These edges need special attention.
If we merge the resulting forms about some subgraphs for
$B_0$ and $B_1$ while edge $uv$ appears both of the two subgraphs,
it means we treat $uv$ doubly so that its length is
$\max\{X_{uv}, X'_{uv}\}$,
where $X_{uv}$ and $X'_{uv}$ are independent and 
identically distributed.
This is acceptable only in case the edge length $X_{uv}$
is a static value since $\Pr[\max\{X_{uv}, X'_{uv}\}\le x]=(F_{uv}(x))^2$.
Instead, we need a definition so that
each random length edge belongs to exactly one bag-subgraph.
Here, we set $\vertex{G_i}=B_i$ and
\begin{align*}
  \edge{G_i}\defeq\{uv\in \edge{G} | \{u,v\}\subseteq B_i ~~\text{and}~~ \forall h\in \Anc(i),~~ \{u,v\}\not\subseteq B_h \},
\end{align*}
where $\Anc(i)$ is the index set of the ancestors of $B_i$ in ${\cal T}(G)$
excluding $i$ itself.
Each edge is given to a bag-subgraph closest to the root.
We call this rule {\em the ancestor-first rule}.
Later, we need the following.
\begin{proposition}
  \label{proposition:separation}
  For any three bags $B_h,B_i$, and $B_j$ where $i\in\Ch(h)$ and
  $j\in \Sub(i)\setminus\{i\}$,
  we have
  $\{uv,vu\}\cap \edge{G}=\emptyset$
  if $u\in B_h\setminus B_i$ and
  $v\in B_j\setminus B_i$.
\end{proposition}
\begin{proof}
  Assume for contradiction that $uv\in \edge{G}$ so that
  $u\in B_h\setminus B_i$ and
  $v\in B_j\setminus B_i~(j\in \Sub(i)\setminus\{i\})$.
  By TD2, 
  there exists $B^\ast$ such that $\{u,v\}\subseteq B^\ast$.
  Let us prove that (1) there is no path from $B_i$ to $B^\ast$ via $B_h$
  and (2) $B^\ast$ is not a descendant of $B_i$.
  To prove (1), suppose for contradiction that there is a path
  from $B_i$ to $B^\ast$ via $B_h$ on ${\mathcal T}(G)$.
  This is a contradiction since 
  $v\in B^\ast\cap B_j\subseteq B_i$ due to TD3
  but we assume $v\in B_j\setminus B_i$.
  Similarly, we prove (2). Suppose for contradiction that
  $B^\ast$ is a descendant of $B_i$.
  Then, $B^\ast\cap B_h\subseteq B_i$ by TD3. 
  Since we have $u\in B^\ast\cap B_h$, this
  is a contradiction to the assumption $u\in B_h\setminus B_i$.
  Thus, we have the claim.
\end{proof}

We consider the union and the intersection of subgraphs.
By the union graph, we merge the partial results for
the bag-subgraphs into $\Pr[X_{\rm MAX}\le x]$. 
By the intersection graph, we define the source and the terminal
of subgraphs.
\begin{definition}
  Let $G'$ and $G''$ be two subgraphs of $G$.
  The {\em union graph} of $G'$ and $G''$ is
  $G'\cup G''=(\vertex{G'}\cup \vertex{G''}, \edge{G'}\cup \edge{G''})$.
  The {\em intersection graph} of $G'$ and $G''$ is 
  $G'\cap G''=(\vertex{G'}\cap \vertex{G''}, \edge{G'}\cap \edge{G''})$.
\end{definition}

We define the sources and the terminals of subgraph $G'$ of $G$ as 
follows;
they are not determined just by in/out-degrees but also
utilizing the incoming/outgoing paths from/to the outside of $G'$.
The sources and the terminals are defined
not just for bag-subgraphs but for subgraphs of $G$.
This is because we consider the union of multiple bag-subgraphs.
\begin{definition}
\label{def:sourceterminal_subgraph}
Let $G'$ be a subgraph of $G$.
A vertex $v\in \vertex{G'}$ is a {\em source of subgraph} $G'$
if there exists an $s-t$ path $\pi$ of $G$ such that
$v\in \vertex{G'}\cap \vertex{\pi}$ has no incoming edge
in graph $G'\cap \pi$.
Similarly, $v$ is a {\em terminal of subgraph} $G'$
if there exists an $s-t$ path $\pi$ of $G$ such that
$v\in\vertex{G'}\cap \vertex{\pi}$ has no outgoing edge 
in $G'\cap \pi$.
If $v\in \vertex{G'}$ is neither a source nor a terminal,
$v$ is an {\em internal vertex of subgraph} $G'$.
\end{definition}
We have this definition
for managing
the order of executing the integrals.
Consider applying Theorem~\ref{th:multiplesourceterminal}
to subgraph $G'$ by simply replacing $G$ in the formula by $G'$.
If $v$ is a source/terminal, $v$ of a subgraph $G'$,
there two possibilities: (1) $v$ is a source/terminal of $G$,
meaning that we do not have to execute the integral with respect to $z_v$;
or (2) $v$ has an incoming/outgoing edge from/to the outside of $G'$,
meaning that there is at least one missing edge. 
Conversely, since $z_v$ does not appear in $F_{uw}(z_u-z_w)$ for
$v\neq u$ and $v\neq w$, we can correctly execute the integral
with respect to $z_v$ if $v$ is not
a source/terminal in the above sense.
Thus, under the definition of the source/terminal of the subgraph,
we set
$\LPpdf{\phi}{G'}{\Vec{z}}{S}{T}\defeq\int_{\mathbb{R}^{\internal{G'}}}\prod_{v\in\internal{G'}}\LPpdf{f}{v}{\Vec{z}}{v}{\Suc(v)}{\rm d}\Vec{z}[\internal{G'}]$.
For bag-subgraphs, we write $S_i\defeq\Src{G_i}$ and  $T_i\defeq\Term{G_i}$
in the following.

There may exist a vertex
$v\in S_i\cap T_i$. 
Here, $v$ may be connected to some paths from/to other bag-subgraphs.
Since such $v$ causes some technical issues,
we will explain how to avoid this case later.
Furthermore, we define the following.
\begin{definition}
  \label{definition:subtreesubgraph}
  The {\em subtree-subgraph} $D_i$ of $B_i$
  is the union of all bag-subgraphs for $j\in \Sub(i)$.
  That is, $D_i=\bigcup_{j\in\Sub(i)}G_j$.
  The {\em uncapped subtree-subgraph} $U_i$ of $B_i$
  is the union of all bag-subgraphs for $j\in \Sub(i)\setminus\{i\}$.
  That is, $U_i=\bigcup_{j\in\Sub(i)\setminus\{i\}}G_j$.
\end{definition}

\subsection{Definite Integrals in Tree Decomposition}
We here describe the longest path length distribution function
in $G_i$ and then
how to
merge them.
The idea shares some commonalities
with the proof of
Theorem \ref{th:multiplesourceterminal}, where
we considered a convolution of $\LPpdf{\phi}{\removeOE{G}{v}}{\Vec{z}}{S}{T}$
and $\LPpdf{f}{v}{\Vec{z}}{v}{\Suc(v)}$
for the topologically last internal vertex $v$.
This corresponds to the union of $\removeOE{G}{v}$ and a graph
given by $v$ and its outgoing edges.
To generalize this argument, 
we consider subsets of the sources and the terminals in $G_i$ and $U_i$
that turn into the internal vertices of $G_i\cup U_i$.
\begin{definition}
Let $B_h$ be the parent of $B_i$ and $j\in \Ch(i)$.
We define 
$\jSrc{G_i}=S_i\cap \Term{U_i}\setminus B_h$ and
$\jTerm{G_i}=T_i\cap \Src{U_i}\setminus B_h$.
For short, we set $\JS_i=\jSrc{G_i}$,
$\JT_i=\jTerm{G_i}$, and $J_i=\JS_i\cup \JT_i$.
\end{definition}
Later, 
we will prove that the vertices in $J_i$ are internal vertices
in $G_i\cup U_i$.
Here $J_i$ does not include the vertices
in $B_h$ since there may be
more incoming/outgoing edges 
from/to the vertices in $B_h$.

We assume one important property of the bag-subgraphs.
\begin{definition}
  \label{definition:separated_tree_decomposition}
  A {\em separated tree decomposition} is a tree decomposition
  ${\cal T}(G)=({\cal B},{\cal A})$
  so that 
  we have $\Src{G_i}\cap \Term{G_i}=\emptyset$
  for all $B_i\in {\cal B}$.
\end{definition}
For general $G$,
we construct a separated tree decomposition as follows.
We define $G^\ast$ by 
$\vertex{G^\ast}=\{v^+,v^-|v\in \vertex{G}\}$
and
$\edge{G^\ast}=E_E\cup E_V$,
where $E_E=\{u^-v^+|uv\in \edge{G_i}\}$ and
$E_V=\{v^+v^-|v\in \vertex{G_i}\}$.
The length of edge $u^-v^+\in E_E$ is equal to
the length $X_{uv}$ of $uv\in \edge{G_i}$.
All the edges in $E_V$ have static length $0$.
The longest path length in $G^\ast$ is equal to that of $G$.

We define the bag-subgraph $G^\ast_i$ of $G^\ast$
using the bag-subgraphs $G_i$'s for the original tree decomposition.
Here, $G^\ast_i$ has vertex set
$\vertex{G^\ast_i}=\{v^+,v^-|v\in \vertex{G_i}\}$ and
edge set $\edge{G^\ast_i}=\{u^-v^+|uv\in \edge{G_i}\}\cup \{v^+v^-|v\in \vertex{G_i}\}$.
Then, each $G^\ast_i$ for $B_i\in {\cal B}$ satisfies
$\Src{G^\ast_i}\cap \Term{G^\ast_i}=\emptyset$
since any plus vertex $v^+\in \vertex{G_i^\ast}$ is not connected
to any outgoing edge to the outside of $G^\ast_i$;
any minus vertex $v^-\in \vertex{G_i^\ast}$ is not connected to
any incoming edge from the outside of $G^\ast_i$.
Though edges in $E_V$ work as multi-edges 
in the union of bag-subgraphs,
this does not change the longest path length since
these edges have static length $0$.
The following is clear by the construction.
\begin{proposition}
  \label{proposition:EE}
  {\rm (a)} For any $uv,vw\in \edge{G^\ast_i}$, we have either (1) $uv\in E_E$ and $vw\not\in E_E$, or (2) $uv\not\in E_E$ and $vw\in E_E$.
  {\rm (b)} If $v\in \Src{G_i^\ast}$, we have $uv \in E_E$ for any $u\in \vertex{G}$.
  {\rm (c)} If $v\in \Term{G_u^\ast}$, we have $vw\in E_E$ for any $w\in \vertex{G}$. 
  {\rm (d)} If $e\in E_E$, there exists exactly one bag-subgraph $G^\ast_i$ such that $e\in \edge{G_i^\ast}$.
\end{proposition}

\begin{figure}[ht]
  \begin{center}
    \includegraphics[clip,height=2.5cm]{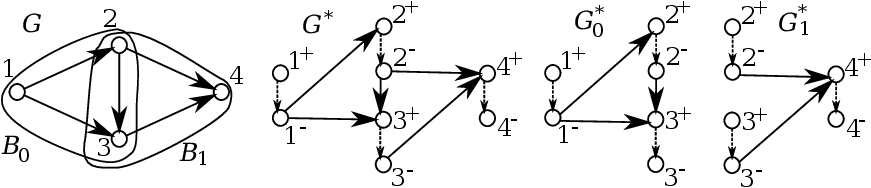}
  \end{center}
  \vspace*{-5mm}
  \caption{Construction of a separated tree decomposition.}
  \label{figure:separated_tree_decomposition}
\end{figure}

As an example, see Fig.~\ref{figure:separated_tree_decomposition}.
The tree decomposition has two bags, $B_0=\{1,2,3\}$ and $B_1=\{2,3,4\}$.
Notice that $2\in S_1\cap T_1$ since,
for an $s-t$ path $\pi$ such that $\edge{\pi}=\{12,23,34\}$,
vertex $2$ has no edge in $\edge{G_1}\cap\edge{\pi}$
by the ancestor-first rule.
In the bag-subgraphs,
we set $23\in \edge{G_0}$ but $23\not\in \edge{G_1}$
by the ancestor-first rule.
Therefore, $2^-3^+\in\edge{G^\ast_0}$ and $2^-3^+\not\in\edge{G^\ast_1}$.
However, we put edges
$2^+2^-$ and $3^+3^-$ in both $\edge{G^\ast_0}$ and $\edge{G^\ast_1}$.
The edges in $E_V$ (in this case, $2^+2^-$ and $3^+3^-$)
work as multi-edges when we merge $G^\ast_0$ and $G^\ast_1$
without changing the longest path length.

Here, $G^\ast_i$ has at most $2k+2$ vertices 
(remember that $|B_i|\le k+1$),
and the width of the tree decomposition is $2k+1$.
Therefore, we deal with the general treewidth $k$ graph
by increasing the
treewidth to at most $2k+1$.

We prove the important properties by assuming that $G$
has a separated tree decomposition constructed by the
above.
\begin{proposition}
  \label{proposition:nointernalintersection}
  For any $i\ne j$, we have $\internal{G_i}\cap \internal{G_j}=\emptyset$.
\end{proposition}
\begin{proof}
  Suppose for contradiction that $v\in\internal{G_i}\cap\internal{G_j}$.
  By definition, $v$ has incoming/outgoing edges
  $e_{1}(\pi),e_{2}(\pi) \in \edge{G_i}\cap \edge{G_j}\cap \edge{\pi}$
  for any $s-t$ path $\pi$.
  Since either one of $e_1(\pi)$ or $e_2(\pi)$ must be in $E_E$
  by (a) of Proposition~\ref{proposition:EE},
  let $e\in \{e_1(\pi),e_2(\pi)\}$ satisfy $e\in E_E$.

  However, we prove that such $e$ does not exist.
  If $e\in \edge{G_i}\cap \edge{G_j}~(e=uv,~\text{or}~e=vu)$,
  $B_i$ and $B_j$ have a common ancestor $B_h$ such that 
  $\{u,v\}\subseteq B_h$ by TD2 and TD3.
  This contradicts (b) and (d) of Proposition~\ref{proposition:EE}.
\end{proof}
\begin{proposition}
  \label{proposition:lrSSTT}
  Let $\Ch(i)=\{\ell,r\}$.
  Given a separated tree decomposition,
  $\Src{D_\ell}\cap\Term{D_r}=\Term{D_\ell}\cap\Src{D_r}=\emptyset$.
\end{proposition}
\begin{proof}
By symmetry, it suffices to prove $\Src{D_\ell}\cap\Term{D_r}=\emptyset$.
Suppose for contradiction that $v\in \Src{D_\ell}\cap\Term{D_r}$.
Since there is no isolated vertex, 
we assume that $v$ has at least one edge $e=uv$ or $e'=vw$.
In the following, however, we prove $e$ does not exist.
The proof of the non-existence of $e'$ is symmetry,
which proves the claim.

We prove $e\not\in\edge{G_i}$.
If $e=uv\in\edge{G_i}$,
we have $e\not\in\edge{D_r}$ by (b) and (d) of Proposition~\ref{proposition:EE}.
Then, we have $v\in\Src{D_r}$
by Definition~{\ref{def:sourceterminal_subgraph}},
meaning that $v\in \Src{D_r}\cap\Term{D_r}$.
Notice that $v\in B_r$ since, by TD3,
$v\in B_j\cap B_{j'}\subseteq B_r$
for any $j\in\Sub(\ell)$ and $j'\in \Sub(r)$
such that $v\in B_j\cap B_{j'}$.
Thus, $v\in S_r\cap T_r$,
contradicting Definition~\ref{definition:separated_tree_decomposition}.

We next prove $e\not\in\edge{D_\ell}\cup\edge{D_r}$.
If $e=uv\in\edge{D_\ell}$ and $v$ has
no outgoing edge in $D_r$, we have that $v$ is an isolated vertex
in $D_r$. This is not possible by the construction.
Note that a contradiction occurs
if $v$ has $e=uv\in\edge{D_\ell}$
and $e'=vw\in\edge{D_r}$ at a time;
we have $e\in E_E$ and $e'\in E_E$
by (b) and (c) of Proposition~\ref{proposition:EE},
but then, either $e=uv$ or $e'=vw$ is not in $E_E$
by (a) of Proposition~\ref{proposition:EE}.
Therefore, $e\not\in\edge{D_\ell}$.
We can prove $e\not\in\edge{D_r}$ by symmetry.

Now, we have $e\not\in\edge{G_i}\cup\edge{D_\ell}\cup\edge{D_r}=\edge{D_i}$,
implying that $v$ does not have any incoming edge $e$.
The rest of the proof is symmetry.
\end{proof}
\begin{proposition}
  \label{proposition:S_i}
  Given a separated tree decomposition,
  we have
  $S_i\cap \Src{U_i}=T_i\cap \Term{U_i}=\emptyset$.
\end{proposition}
\begin{proof}
  We first claim $\Src{U_i}\cap \Term{U_i}=\emptyset$
  for any bag $B_i\in{\cal B}$.
  The proof is the induction on $i$ from the leaves to the root.
  As for the base, let $B_i$ be a leaf bag in ${\cal T}(G)$.
  Since $U_i=G_i$, we have $\Src{U_i}\cap \Term{U_i}=\emptyset$
  by Definition~\ref{definition:separated_tree_decomposition}.
  We proceed to the induction step.
  Consider the case $\Ch(i)=\{j\}$,
  implying $U_i=D_j=G_j\cup U_j$.
  By the induction hypothesis and
  Definition~\ref{definition:separated_tree_decomposition},
  we have $\Src{U_i}\cap\Term{U_i}=\emptyset$
  if $v\not\in\Src{D_j}\cap\Term{D_j}$
  for $v\in(\Src{U_j}\cap \Term{G_j})\cup(\Term{U_j}\cap\Src{G_j})$.
  In case $v\in \Src{U_j}\cap\Term{G_j}$,
  there exists an $s-t$ path $\pi$ of $G$ so that
  $v$ has in-degree (resp. out-degree) $0$
  in $U_j\cap\pi$ (resp. $G_j\cap\pi$)
  by Definition~\ref{def:sourceterminal_subgraph}.
  Since $v$ has in-degree and out-degree one in $D_j\cap\pi$, 
  we have $v\not\in\Src{D_j}\cap\Term{D_j}$.
  The case $v\in\Term{U_j}\cap\Src{G_j}$ is similar.
  Next, consider the case $\Ch(i)=\{\ell,r\}$.
  Since $\Src{D_\ell}\cap \Term{D_r}=\Term{D_\ell}\cap\Src{D_r}=\emptyset$
  by Proposition~\ref{proposition:lrSSTT},
  we have $\Src{U_i}\cap\Term{U_i}=\emptyset$.
  Thus, $\Src{U_i}\cap\Term{U_i}=\emptyset$ for any bag $B_i\in{\cal B}$.

  To prove $S_i\cap\Src{U_i}=\emptyset$, assume for contradiction 
  $v\in S_i\cap \Src{U_i}$. Then, $v$ has an edge $vw\in\edge{D_i}$
  since otherwise $v\in T_i\cup \Term{U_i}$, a contradiction.
  Suppose that $w\in B_i$.
  Notice that $vw\not\in \edge{U_i}$ by (c) and (d) of
  Proposition~\ref{proposition:EE},
  implying $v\in \Term{U_i}$ by Definition~\ref{def:sourceterminal_subgraph},
  a contradiction.
  Suppose that $w\in \vertex{U_i}\setminus B_i$.
  Then, we have $v\in T_i$, a contradiction.
  The proof $T_i\cap \Term{U_i}=\emptyset$ is symmetry.
\end{proof}


\begin{proposition}
  \label{proposition:newinternalvertex}
  Given a separated tree decomposition,
  %
  we have $J_i=\internal{D_i}\setminus (\internal{G_i}\cup \internal{U_i})$.
  That is, any $v\in J_i$ is not an internal vertex of $G_i$ nor $U_i$,
  and $v$ is an internal vertex of $D_i=G_i\cup U_i$. 
\end{proposition}
\begin{proof}
  Let $v\in J_i=\JS_i\cup \JT_i$.
  Since $v\not\in \internal{G_i}\cup\internal{U_i}$ by definition,
  we prove that $v\in\internal{D_i}$.
  Let $B_h$ be the parent of $B_i$.
  By Proposition~\ref{proposition:separation} and
  that $v\in B_i\setminus B_h$,
  there is no edge between $v$ and
  vertex $u\in B_g$ for any $g\in \Anc(i)$.
  Since all the neighbor of $v$ is in $D_i$,
  we have $v\in \internal{D_i}$.
\end{proof}

We are ready to explain how we compute
$\LPpdf{\phi}{D_i}{\Vec{z}}{S}{T}$
assuming $\LPpdf{\phi}{G_i}{\Vec{z}}{S}{T}$ and
$\LPpdf{\phi}{U_i}{\Vec{z}}{S}{T}$.
Consider an example where
any path that goes from $G_i$ to $U_i$ includes $v$.
Since we subtract $z_v$
from the length of the longest path ending at $v$ in $G_i$, 
the product $\LPpdf{\phi}{G_i}{\Vec{z}}{S}{T}\LPpdf{\phi}{U_i}{\Vec{z}}{S}{T}$
is the density of the longest path length in $G_i\cup U_i$
in case the longest path length from the source of $G_i$ to
$v$ is equal to 
$z_v$.
Then, we obtain $\LPpdf{\phi}{D_i}{\Vec{z}}{S}{T}$
by introducing an integral with respect to $z_v$.
In general case, we have one integral for each $v\in J_i$.
\begin{lemma}
  \label{lemma:phi_increment}
  Given a separated tree decomposition,
  we have
\begin{align}
  \LPpdf{\phi}{D_i}{\Vec{z}}{S}{T}&=\int_{\mathbb{R}^{J_i}}\hspace*{-2mm} \LPpdf{\phi}{G_i}{\Vec{z}}{S}{T}\LPpdf{\phi}{U_i}{\Vec{z}}{S}{T}{\rm d}\Vec{z}[J_i], ~~\text{and}\label{form:phiDi}\\
  \LPpdf{\Phi}{D_i}{\Vec{z}}{S}{T}&=\int_{\mathbb{R}^{J_i}}\left(\vecdiff{\JS_i}\LPpdf{\Phi}{G_i}{\Vec{z}}{S}{T}\right)\left(\vecdiff{\JT_i}\LPpdf{\Phi}{U_i}{\Vec{z}}{S}{T}\right){\rm d}\Vec{z}[J_i], \label{form:PhiDi}
\end{align}
where $\LPpdf{\Phi}{U_i}{\Vec{z}}{S}{T}=\prod_{j\in \Ch(i)} \LPpdf{\Phi}{D_j}{\Vec{z}}{S}{T}$.
\end{lemma}
We have $\Pr[X_{\rm MAX}\le x]=\LPpdf{\Phi}{D_0}{x\Vec{\sigma}}{S}{T}$, where $\Vec{\sigma}[s]=1$ for $s\in \Src{G}$ and $\Vec{\sigma}[v]=0$ for $v\not\in\Src{G}$.
\begin{proof}
Notice that $\LPpdf{\Phi}{U_i}{\Vec{z}}{S}{T}=\prod_{j\in \Ch(i)} \LPpdf{\Phi}{D_j}{\Vec{z}}{S}{T}$ by Proposition \ref{proposition:lrSSTT}.
We first prove (\ref{form:phiDi}).
To be careful about the difference of
the definitions of sources/terminals for the entire $G$
and its subgraphs,
we apply Theorem~\ref{th:multiplesourceterminal}
to $D_i$ by a minor modification.
For each source $s$ (resp. terminal $t$), consider a new source vertex $v_s$ (resp. a new terminal vertex $v_t$) with edge $v_ss$ (resp. $tv_t$) with static length $0$.
Then, 
$\LPpdf{\phi}{D_i}{\Vec{z}}{S}{T}=\int_{\mathbb{R}^{\internal{D_i}}}\prod_{u\in \Src{D_i}\cup\internal{D_i}}\LPpdf{f}{u}{\Vec{z}}{u}{\Suc(u)} {\rm d}\Vec{z}[\internal{D_i}]$.

We group the factors of the integrand into two.
Let $W_1 = S_i\cup\internal{G_i}$ and
$W_2=\Src{U_i}\cup \internal{U_i}$.
Since
$\internal{G_i}\cap \internal{U_i}=\emptyset$ and
$S_i\cap \Src{U_i}=\emptyset$ by Propositions~\ref{proposition:nointernalintersection} and \ref{proposition:S_i},
we have $W_1\cap W_2=\emptyset$.
Let $Q(\Vec{z},W)=\prod_{u\in W}\LPpdf{f}{u}{\Vec{z}}{u}{\Suc(u)}$.
We have
\begin{align*}
  \LPpdf{\phi}{D_i}{\Vec{z}}{S}{T}=\int_{\mathbb{R}^{\internal{D_i}}}Q(\Vec{z},W_1)Q(z,W_2){\rm d}\Vec{z}[\internal{D_i}].
\end{align*}

Then, by Theorem~\ref{th:multiplesourceterminal},
we have
$\LPpdf{\phi}{G_i}{\Vec{z}}{S}{T}=\int_{\mathbb{R}^{\internal{G_i}}} Q(\Vec{z},W_1){\rm d}\Vec{z}[\internal{G_i}]$
and $\LPpdf{\phi}{U_i}{\Vec{z}}{S}{T}=\int_{\mathbb{R}^{\internal{U_i}}} Q(\Vec{z},W_2){\rm d}\Vec{z}[\internal{U_i}]$. Now, we have that
\begin{align*}
  \LPpdf{\phi}{D_i}{\Vec{z}}{S}{T}=\int_{\mathbb{R}^{\internal{D_i}\setminus (\internal{G_i}\cup \internal{U_i})}} \LPpdf{\phi}{G_i}{\Vec{z}}{S}{T}\LPpdf{\phi}{U_i}{\Vec{z}}{S}{T}{\rm d}\Vec{z}[\internal{D_i}].
\end{align*}
Since $\internal{D_i}\setminus(\internal{G_i} \cup \internal{U_i})=J_i$ by Proposition~\ref{proposition:newinternalvertex}, we have (\ref{form:phiDi}).

We have (\ref{form:PhiDi}) by integrating (\ref{form:phiDi}) with respect to $\Vec{z}[\Src{D_i}]$.
\end{proof}

Let $\Phi(x)$ and $\Psi(x)$ be two probability distribution functions whose
derivatives are $\phi(x)$ and $\psi(x)$ respectively.
Since integration by parts leads to that
$\int_{\mathbb{R}}\Phi(s)\psi(x-s){\rm d}s=\int_{\mathbb{R}}\phi(s)\Psi(x-s){\rm d}x$,
we extend this idea as follows.
\begin{proposition}
  \label{proposition:switch}
  Given a separated tree decomposition, we have
  \begin{align}
    \LPpdf{\Phi}{D_i}{\Vec{z}}{S}{T}&=\int_{\mathbb{R}^{J_i}}(-1)^{|\JT_i|}\left(\vecdiff{J_i}\LPpdf{\Phi}{G_i}{\Vec{z}}{S}{T}\right)~\LPpdf{\Phi}{U_i}{\Vec{z}}{S}{T}{\rm d}\Vec{z}[J_i] \label{form:switch_earlier}\\
    &=\int_{\mathbb{R}^{J_i}}\LPpdf{\Phi}{G_i}{\Vec{z}}{S}{T}~(-1)^{|\JS_i|}\left(\vecdiff{J_i}\LPpdf{\Phi}{U_i}{\Vec{z}}{S}{T}\right){\rm d}\Vec{z}[J_i]. \label{form:switch_latter}
  \end{align}
\end{proposition}
\begin{proof}
  Applying integration by parts with respect to $\Vec{z}[t]$ for
  $t\in\JT_i$ to (\ref{form:PhiDi}) of Lemma~\ref{lemma:phi_increment},
  we have that 
  \begin{align*}
    \LPpdf{\Phi}{D_i}{\Vec{z}}{S}{T}=-\int_{\mathbb{R}^{J_i}}\left(\vecdiff{\JS_i\cup\{t\}}\LPpdf{\Phi}{G_i}{\Vec{z}}{S}{T}\right)~\left(\vecdiff{\JT_i\setminus\{t\}}\LPpdf{\Phi}{U_i}{\Vec{z}}{S}{T}\right){\rm d}\Vec{z}[J_i],
  \end{align*}
  since  $(\vecdiff{\JS_i}\LPpdf{\Phi}{G_i}{\Vec{z}}{S}{T})(\vecdiff{\JT_i\setminus\{t\}}\LPpdf{\Phi}{U_i}{\Vec{z}}{S}{T})$ converges to $0$ by definition
  when $\Vec{z}[t]\rightarrow \pm\infty$.
  Now, we repeat the integration by parts with respect to
  all $\Vec{z}[\JT_i]$ components, which leads to (\ref{form:switch_earlier}).
  The proof for (\ref{form:switch_latter}) is symmetry.
\end{proof}

\section{Approximation for DAGs with Uniformly Distributed Edge Lengths}
The computation
in the previous section 
may be slow when the edge lengths are uniformly distributed.
The obstacle is that there may be 
exponentially many breakpoints in the derivative of some order of
$\LPpdf{\Phi}{D_0}{x\Vec{\sigma}}{S}{T}$ with respect to $x$,
where $\Vec{\sigma}[s]=1$ for $s\in \Src{G}$ and $\Vec{\sigma}[v]=0$
for $v\not\in\Src{G}$.
Each component of a vector $\Vec{a}\in \mathbb{Z}_{> 0}^{E}$ is a parameter of
a uniform distribution for $e\in E$.
By considering a uniform random vector $\Vec{X} \in [0,1]^{E}$,
the length of edge $e\in E$ is $\Vec{a}[e]\Vec{X}[e]$.
Then, we consider computing the probability
$\Pr[X_{\rm MAX}\le x]=\LPpdf{\Phi}{G}{x\Vec{\sigma}}{S}{T}={\rm Vol}(K_G(\Vec{a},x))$.
Though the problem of computing the exact value
is a $\#$P-hard problem under
this definition of the edge lengths, we show 
a deterministic FPTAS.
We, in a sense, generalize the algorithm in~\cite{AK2016} for 
the volume of multiple constraint knapsack polytope
so that 
we can compute the volume of the hypercube $[0,1]^E$
truncated by exponentially many halfspaces.

\subsection{Idea of Our Algorithm}
The key is the discretization of the dummy variables
to replace the integrals with the discrete summations as in~\cite{AK2016}.
Intuitively,
the hardness of the exact computation is due to
the exponentially many breakpoints of $\Pr[X_{\rm MAX}\le x]$.
We apply a kind of rounding technique. We
pick up a relatively small number of values
as a staircase function
(i.e., a piecewise constant function with
breakpoints lined up at equal intervals)
so that the exact value is a lower bound on our approximation.
The computation of the staircase function
is efficient at the leaves of the tree decomposition.
We then compute the union of the subtrees
at the cost of mildly increasing the error.

The advantage of the staircase approximation is that
we have $\LPpdf{\Phi}{D_i}{x\Vec{\sigma}}{S}{T}$ as the lower bound
and $\LPpdf{\Phi}{D_i}{(x+h)\Vec{\sigma}}{S}{T}$ as the upper bound
for some $h$.
We call $h$ ``horizontal error.''
We prove that the horizontal error increases
by addition as the computation advances from the
leaves toward the root.
In the analysis, 
we introduce the idea of vertex discount, which is
the horizontal error associated with every source of $G_i$.
We bound the total horizontal error for
$\LPpdf{\Phi}{D_0}{\Vec{z}}{S}{T}$ at the root.
Then, it leads to bounding the multiplicative error
by examining $\Pr[X_{\rm MAX}\le x]$
as an $m$-dimensional volume.
Later, we call the multiplicative approximation ratio 
``vertical'' approximation ratio.
The notions of horizontal error and vertical error are analogous to
plotting $\Pr[X_{\rm MAX}\le x]$ and its approximation
in two dimensional space;
the horizontal axis and the vertical axis correspond to
the argument $x\in\mathbb{R}$ and the value of 
$\Pr[X_{\rm MAX}\le x]$, respectively.

\subsection{Detail of Our Algorithm}
We assume that $\Vec{a}[e]\le x$ for any $e\in \edge{G}$.
In case we have $\Vec{a}[e]> x$, we replace $\Vec{a}[e]$ by $x$.
We can obtain ${\rm Vol}(K_G(\Vec{a},x))$ from ${\rm Vol}(K_G(\Vec{a}',x))$.
\begin{proposition}
  \label{proposition:atmostx}
  Let $E(x)$ be the set of edges of $G$ such that $\Vec{a}[e]>x$
  for $e\in E(x)$.
  Let $\Vec{a}'[e]\in (\mathbb{Z}_{> 0}\cup \{x\})^E$.
  We set $\Vec{a}'[e]=\Vec{a}[e]$ if $\Vec{a}[e]\le x$ and
  $\Vec{a}'[e]=x$ if $\Vec{a}[e]>x$.
  Then, we have 
  $K_G(\Vec{a},x)=K_G(\Vec{a}',x)\prod_{e\in E(x)}\frac{x}{\Vec{a}[e]}$.
\end{proposition}
\begin{proof}
  Notice that the claim follows from 
  Theorem~\ref{th:multiplesourceterminal} if $y=z_u-z_v\le x$
  for all $uv\in \edge{G}$.
  in the integrand of Theorem~\ref{th:multiplesourceterminal}.
  In that case, replacing $\Vec{a}[e]$ by $\Vec{a}'[e]$
  means that we multiply the uniform density $f_{uv}(y)$ of each edge
  $uv\in E(x)$ by $\Vec{a}[uv]/x$ for $0\le y \le x$. 

  In the following,
  we prove $z_u-z_v\le x$ for all $uv\in \edge{G}$
  or the integrand is $0$ by induction on $v$.
  As for the base case, consider the case $v\in\Suc(s)$ for any
  source vertex $s$.
  We have $z_s=x$ and $z_v\ge 0$ for any $v\in\vertex{G}$ by definition,
  implying $z_s-z_v\le x$ for all $sv\in\edge{G}$.
  In the induction step, we assume that each topologically earlier
  vertex $u$ of $v$ satisfies that $z_u\le x$.
  Suppose that $z_v-z_w> x$ for $v\not\in\Src{G}$
  and its outgoing edge $vw\in \edge{G}$. Since this
  implies $z_v>x$, there exists $uv\in\edge{G}$ so that $z_u-z_v<0$,
  by the induction hypothesis $z_u\le x$.
  Since the above arguments imply $F_{uv}(z_u-z_v)=0$,
  the entire integrand is $0$
  if edge $vw\in\edge{G}$ exists satisfying $z_v-z_w>x$.
\end{proof}

We approximate $\LPpdf{\Phi}{G_i}{\Vec{z}}{S}{T}$ by
a staircase function $\LPpdfx{\Lambda}{G_i}{\Vec{z}}{S}{T}{M}$
using parameter $M$.
Assuming a width $k$ separated tree decomposition, 
we prove that
$M=\lceil 2(k+1)bmn/\epsilon\rceil$ is
sufficient to bound the approximation ratio at most $1+\epsilon$.
We compute
the discrete version of (\ref{form:PhiDi}) in Lemma~\ref{lemma:phi_increment}
from the leaves to the root so that we have the approximation of 
$\LPpdf{\Phi}{D_0}{\Vec{z}}{S}{T}$.

For any value of $\Vec{z}[S_i\cup T_i]$, we compute
$\LPpdfx{\Lambda}{G_i}{\Vec{z}}{S}{T}{M}$ by counting.
We divide $P_i\defeq [0,1]^{\edge{G_i}}$ into $M^{|\edge{G_i}|}$
small orthogonal hypercubes whose diagonal
extreme points are
$\frac{1}{M}\Vec{g}_E$ and $\frac{1}{M}(\Vec{g}_E+\Vec{1})$
for $\Vec{g}_E\in \{0,1,\dots,M-1\}^{\edge{G_i}}$.
We call the small hypercubes {\em cells}.
Given a value of $\Vec{z}[S_i\cup T_i]$,
we count the number $\LPpdfx{N}{G_i}{\Vec{z}}{S}{T}{M}$ of cells intersecting
\begin{align*}
  \LPpdf{K'}{G_i}{\Vec{z}}{S}{T}\defeq\left\{\Vec{x}\in P_i\left| \bigwedge_{\pi\in \Pi_{G_i}}\sum_{uv\in \edge{\pi}} \Vec{a}[uv] \Vec{x}[uv]\le \Vec{z}[\src{\pi}]-\Vec{z}[\term{\pi}] \right.\right\}.
\end{align*}

We achieve the counting by considering the following static
longest path problem in $G_i$.
We make another DAG $G_i^0$ from $G_i$ where $\vertex{G_i^0}=B_i\cup \{s_0,t_0\}$
and $\edge{G_i^0}=\edge{G_i}\cup \bigcup_{s\in S_i}\{s_0s\}\cup \bigcup_{t\in T_i}\{tt_0\}$.
The length of $uv\in \edge{G_i}$ is $\Vec{a}[uv]\Vec{x}[uv]$.
For $s\in S_i$, the length of $s_0s$ is $-\Vec{z}[s]$.
Similarly, for $t\in T_i$, the length of $tt_0$ is $\Vec{z}[t]$.
Let $\Vec{g}$ be the discretized vector
$\Vec{g}[v]=\left\lceil\frac{M}{x}\Vec{z}[v]\right\rceil$ for $v\in S_i$ and
$\Vec{g}[v]=\left\lfloor \frac{M}{x}\Vec{z}[v]\right\rfloor$ for $v\in T_i$.
Then, $\LPpdfx{N}{G_i}{\frac{x}{M}\Vec{g}}{S}{T}{M}$ is the number of gridpoints
$\Vec{x}=\Vec{g}_E/M$ in $P_i$,
where the longest path length
given by $\Vec{x}$ is at most $0$.
To solve a static longest path problem in a DAG,
see, e.g., \cite{I2A}.
We approximate $\LPpdf{\Phi}{G_i}{\Vec{z}}{S}{T}$ by 
\begin{align}
\textstyle\LPpdfx{\Lambda}{G_i}{\Vec{z}}{S}{T}{M}\defeq \LPpdfx{N}{G_i}{\frac{x}{M}\Vec{g}}{S}{T}{M} M^{-|\edge{G_i}|}. \label{form:ABar}
\end{align}
By assuming a separated tree decomposition,
we have that $\LPpdfx{\Lambda}{G_i}{\Vec{z}}{S}{T}{M}$ is
an upper bound of $\LPpdf{\Phi}{G_i}{\Vec{z}}{S}{T}$
since $\LPpdf{\Phi}{G_i}{\Vec{z}}{S}{T}$ is monotonically increasing
(resp. decreasing)
with respect to $\Vec{z}[S_i]$ components
(resp. $\Vec{z}[T_i]$ components).

Let us estimate the running time to compute
$\LPpdfx{N}{G_i}{\frac{x}{M}\Vec{g}}{S}{T}{M}$
for all $\Vec{g}\in \{0,1,\dots,M\}^{S_i\cup T_i}$.
We can solve the longest path problem in $G_i^0$ in
$O(|\vertex{G_i^0}|+|\edge{G_i^0}|)=O(k^2)$ time,
where $k$ is the treewidth of $G$.
There are $M^{|\edge{G_i}|}$ cells 
of $P_i=[0,1]^{\edge{G_i}}$ and
the number of grid points in $P_i$ is at most
$(M+1)^{|\edge{G_i}|}\le (M+1)^{k(k+1)} = O(M^{k^2+k})$.
We observe the following.
\begin{observation}
  \label{observation:A_i}
  Given $M$ and a separated tree decomposition of $G$ with width $k$,
  we can obtain an array of 
  $\LPpdfx{\Lambda}{G_i}{\Vec{z}}{S}{T}{M}$ for all possible input
  $\Vec{z}\in [0,x]^{S_i\cup T_i}$ in
  $O(k^2 M^{k^2+k})$ time.
\end{observation}

We merge $\LPpdfx{\Lambda}{G_i}{\Vec{z}}{S}{T}{M}$'s 
for $B_i\in {\cal B}$ into the approximation of
$\LPpdf{\Phi}{D_0}{\Vec{z}}{S}{T}$ as follows.
We define operator $\Delta(z_v)$
for $z_v=\Vec{z}[v]~(v\in S_i)$ as,
\begin{align*}
  \Delta(z_v)\LPpdfx{\Lambda}{G_i}{\Vec{z}}{S}{T}{M}\defeq \left(\LPpdfx{\Lambda}{G_i}{\Vec{z}^{(v)}}{S}{T}{M}-\LPpdfx{\Lambda}{G_i}{\Vec{z}}{S}{T}{M}\right)\frac{M}{x},
\end{align*}
where $\Vec{z}^{(v)}[v]=\Vec{z}[v]+x/M$  and $\Vec{z}^{(v)}[w]=\Vec{z}$
for any $w\neq v$.
We multiply $M/x$ in order to make the difference operator approximates
the derivative.
Let $S_i=\{u_1,\dots,u_{|S_i|}\}$.
The difference with respect to $\Vec{z}[S_i]$ is
\begin{align*}
  \Delta(\Vec{z}[S_i])\LPpdfx{\Lambda}{G_i}{\Vec{z}}{S}{T}{M}\defeq\Delta(z_{u_1})\cdots \Delta(z_{u_{|S_i|}}) \LPpdfx{\Lambda}{G_i}{\Vec{z}}{S}{T}{M}.
\end{align*} 

Here,  we define
$\LPpdfx{\Lambda}{D_i}{\Vec{z}}{S}{T}{M}$ as
an approximation of $\LPpdf{\Phi}{D_i}{\Vec{z}}{S}{T}$.
For each leaf bag $B_c\in {\cal B}$,
we set $\LPpdfx{\Lambda}{D_c}{\Vec{z}}{S}{T}{M}\defeq \LPpdfx{\Lambda}{G_c}{\Vec{z}}{S}{T}{M}$.
As for $U_i$, we set $\LPpdfx{\Lambda}{U_i}{\Vec{z}}{S}{T}{M}\defeq \prod_{j\in \Ch(i)}\LPpdfx{\Lambda}{D_j}{\Vec{z}}{S}{T}{M}$.
Then, we define $\LPpdfx{\Lambda}{D_i}{\Vec{z}}{S}{T}{M}$ as
\begin{align}
  \LPpdfx{\Lambda}{D_i}{\Vec{z}}{S}{T}{M}&\defeq \sum_{\Vec{g}\in \{0,\dots,M\}^{J_i}}\Delta(\JS_i)\LPpdfx{\Lambda}{G_i}{\Vec{u}}{S}{T}{M}\Delta(\JT_i)\LPpdfx{\Lambda}{U_i}{\Vec{u}}{S}{T}{M}, \label{form:Lambda_i}
\end{align}
where we set $\Vec{u}[v]=\frac{x}{M}\Vec{g}[v]$ for $v\in J_i$, and $\Vec{u}[v]=\Vec{z}[v]$ for $v\not\in J_i$.
Let $\Vec{v}[s]=\frac{x}{M}\Vec{g}[s]$ for $s\in \Src{D_i}$ and $\Vec{v}[v]=\Vec{z}[v]$ for $v\not\in\Src{D_i}$.

Note that (\ref{form:Lambda_i}) is analogous to (\ref{form:PhiDi})
of Lemma~\ref{lemma:phi_increment}.
We approximate (\ref{form:PhiDi}) using 
this discrete sum.
To obtain the values of $\Delta(\JS_i)\LPpdfx{\Lambda}{G_i}{\Vec{z}}{S}{T}{M}$,
we first compute the values of
$\LPpdfx{\Lambda}{G_i}{\Vec{z}}{S}{T}{M}$ at the gridpoints
$\Vec{z}[S_i\cup T_i]=\frac{x}{M}\Vec{g}$ 
for $\Vec{g}\in \{0,\dots,M\}^{S_i\cup T_i}$.
Then, we compute
$\Delta(z_{v_1})\LPpdfx{\Lambda}{G_i}{\frac{x}{M}\Vec{g}}{S}{T}{M}$ 
for all $\Vec{g}\in\{0,\dots,M\}^{S_i\cup T_i}$.
We store the values 
in an array with $(M+1)^{|S_i\cup T_i|}$ elements and compute
$\Delta(z_{v_2})\Delta(z_{v_1})\LPpdfx{\Lambda}{G_i}{\frac{x}{M}\Vec{g}}{S}{T}{M}$
similarly for $v_2\in S_i$.
We repeat computing these differences for $v_1,v_2,\dots,v_{|S_i|}\in S_i$.
When we obtain the values of
$\Delta(\JS_i)\LPpdfx{\Lambda}{G_i}{\frac{x}{M}\Vec{g}}{S}{T}{M}$
for $\Vec{g}\in \{0,\dots,M\}^{S_i\cup T_i}$,
we will have $|S_i|$ arrays
with $(M+1)^{|S_i\cup T_i|}$ elements each.
We similarly obtain
$\Delta(\JT_i)\LPpdfx{\Lambda}{U_i}{\frac{x}{M}\Vec{g}}{S}{T}{M}$.
We have the following.
\begin{observation}  \label{observation:Delta_A_i}
Given a separated tree decomposition,
an array of all possible values of   
$\Delta(\JS_i)\LPpdfx{\Lambda}{G_i}{\frac{x}{M}\Vec{g}}{S}{T}{M}$
and $\Delta(\JT_i)\LPpdfx{\Lambda}{U_i}{\frac{x}{M}\Vec{g}}{S}{T}{M}$
can be computed in $O(M^{k+1})$ time.
\end{observation}

We save the memory space
by considering the following $\Vec{w}$.
Let $B_h$ is the parent of $B_i$ and
$\Vec{g}\in \{0,\dots,M\}^{(\Src{D_i}\cup \Term{D_i})\cap B_h}$.
We compute (\ref{form:Lambda_i}) 
setting $\Vec{w}[v]=\frac{x}{M}\Vec{g}[v]$ for $v\in (\Src{D_i}\cup \Term{D_i})\cap B_h$; $\Vec{w}[s]=x$ for $s\in \Src{D_i}\setminus B_h$ and $\Vec{w}[t]=0$ for $t\in \Term{D_i}\setminus B_h$.
Observe that, by Proposition~\ref{proposition:separation},
$s\in \Src{D_i}\setminus B_h$ and $t\in \Term{D_i}\setminus B_h$
are not going to be connected to any vertex in
the ancestor bag of $B_i$.
Therefore, these $\Vec{w}[s]$'s and $\Vec{w}[t]$'s
are not the dummy variables of the convolutions.
By reducing the number of variables using this replacement,
we can store $\LPpdfx{\Lambda}{D_i}{\Vec{w}}{S}{T}{M}$
as an array with
at most $(M+1)^{|(S_i\cup T_i)\setminus J_i|}\le (M+1)^{k+1}$ elements
since $\LPpdfx{\Lambda}{D_i}{\Vec{w}}{S}{T}{M}$ is a staircase function.
In Step 08 of Algorithm 1 for $\LPpdfx{\Lambda}{D_i}{\Vec{w}}{S}{T}{M}$,
we compute a sum of at most $(M+1)^{k+1}= O(M^{k+1})$ values 
for each value of $\frac{x}{M}\Vec{g}$ for
$\Vec{g}\in \{0,1,\dots,M\}^{J_i}$.
We have the following observation.
\begin{observation}
  \label{observation:lambda_j}
  Assume a separated tree decomposition with width $k$.
  Given $\LPpdfx{\Lambda}{G_i}{\Vec{z}}{S}{T}{M}$ and 
  $\LPpdfx{\Lambda}{U_i}{\Vec{z}}{S}{T}{M}$,
  we can compute an array of all possible values of
  $\LPpdfx{\Lambda}{D_i}{\Vec{w}}{S}{T}{M}$ in $O(M^{2k+2})$ time.
\end{observation}

The following pseudocode shows our algorithm Approx-DAG.
\begin{algorithm}
Approx-DAG($G,\Vec{a},x,{\cal T}(G)$) \\
Input: DAG $G$, edge lengths parameter $\Vec{a}\in \mathbb{Z}_{>0}^{E}$, $x\in\mathbb{R}$, and \\ 
\hspace*{1cm}  binary tree decomposition ${\cal T}(G)$ with bags $B_0,\dots,B_{b-1}\in{\cal B}$; \\
Output: Approximate value of $\Pr[X_{\rm MAX}\le x]$;\\
01. For each bag $B_i\in {\cal B}$ do \\
02. \hspace*{2mm} Compute $\LPpdfx{\Lambda}{G_i}{\frac{x}{M}\Vec{g}}{S}{T}{M}$ for all $\Vec{g}\in\{0,\dots,M\}^{S_i\cup T_i}$ by (\ref{form:ABar}); \\
03. done; \\
04. For each leaf bag $B_c \in {\cal B}$ do\\
05. \hspace*{2mm} Set $\LPpdfx{\Lambda}{D_c}{\frac{x}{M}\Vec{g}}{S}{T}{M}:=\LPpdfx{\Lambda}{G_c}{\Vec{w}}{S}{T}{M}$ for all possible $\Vec{w}$;\\
06. done;\\
07. For each bag $B_i$ from the leaves to the root $B_0$ do \\
08. \hspace*{2mm} Compute $\LPpdfx{\Lambda}{D_i}{\Vec{w}}{S}{T}{M}$ by (\ref{form:Lambda_i}) for all possible $\Vec{w}$;\\
09. done; \\
10. Output $\LPpdfx{\Lambda}{D_0}{x\Vec{\sigma}}{S}{T}{M}$.
\end{algorithm}

Remember that we construct a separated tree decomposition
by increasing the treewidth to at most $2k+1$.
Taking the sum of the running time in Observations \ref{observation:A_i},\ref{observation:Delta_A_i}, and \ref{observation:lambda_j}, we have the following.
\begin{observation}
  \label{observation:runningtime}
  The running time of Approx-DAG is
  $O(k^2 b M^{4k^2+6k+2})$. 
\end{observation}

\subsection{Analysis}
In this section, we prove the following theorem.
Theorem 1 follows from Theorem~\ref{th:runningtime}
by replacing the treewidth $k$ with $2k+1$ for constructing the separated
tree decomposition.
\begin{theorem}
\label{th:runningtime}
Let $G$ be a DAG whose underlying undirected graph has
a separated tree decomposition with width $k$.
Our algorithm outputs 
$V'=\LPpdfx{\Lambda}{D_0}{x\Vec{\sigma}}{S}{T}{M}$ satisfying 
$1\le V'/\LPpdf{\Phi}{D_0}{x\Vec{\sigma}}{S}{T}\le 1+\epsilon$,
where $\Vec{\sigma}[s]=1$ for $s\in \Src{G}$ and $\Vec{\sigma}[v]=0$
for $v\not\in\Src{G}$.
The running time of our algorithm is
$O\left(k^2 n\left(\frac{8(k+1)mn^2}{\epsilon}\right)^{k^2+k}\right)$.
\end{theorem}
We prove $M\ge \lceil 2(k+1)bmn/\epsilon\rceil$ is sufficient. 
Instead of directly evaluating the approximation ratio
$\LPpdfx{\Lambda}{D_0}{x\Vec{\sigma}}{S}{T}{M}/\LPpdf{\Phi}{D_0}{x\Vec{\sigma}}{S}{T}$,
we first seek an upper bound of $\LPpdfx{\Lambda}{D_0}{x\Vec{\sigma}}{S}{T}{M}$
given in the form of $\LPpdf{\Phi}{D_0}{(x+h)\Vec{\sigma}}{S}{T}$
for horizontal error $h$,
which we later prove that $h$ at most $\frac{\epsilon x}{2m}$ is sufficient.
Since $\LPpdf{\Phi}{D_0}{x\Vec{\sigma}}{S}{T}$ is a lower bound on
$\LPpdfx{\Lambda}{D_0}{x\Vec{\sigma}}{S}{T}{M}$,
we prove the vertical approximation ratio
as $\LPpdf{\Phi}{D_0}{(x+h)\Vec{\sigma}}{S}{T}/\LPpdf{\Phi}{D_0}{x\Vec{\sigma}}{S}{T}$,
by showing that this is at most the ratio of two
geometrically similar $m$-dimensional polytopes' volumes.

To evaluate the horizontal error,
we introduce the idea of {\em vertex discount}.
Since we have our error estimation for each bag, 
we must be careful because
the total error estimation depends 
on how $s-t$ paths go through the sources of bag-subgraphs
(Remember Fig.\ref{fig:FoldedPath}).
We explain the effect of the horizontal error in each bag
by associating the vertex discount with
the sources of the corresponding bag-subgraph.
Then, we define the longest path length distribution function
that includes the vertex discount.
\begin{definition}
  For $v\in \vertex{G}$, subgraph $G'$ of $G$, and $W\subseteq \vertex{G}$,
  we define a vector $d_{G'}^W\in\mathbb{R}^{\vertex{G'}}$ as follows.
  If any bag-subgraph $G_i$ is also a subgraph of $G'$,
  we set $d_{G'}^W(v)=(k+1)x/M$ for $v\in W\cap \Src{G_i}$ and
  $d_{G'}^W(v)=kx/M$ for $v\in \Src{G_i}\setminus W$;
  otherwise, $d_{G'}^W(v)=0$.
  We call $d_{G'}^W$ {\em the vertex discount}.
  When $W=\emptyset$, we write $d_{G'}^\emptyset=d_{G'}$.
\end{definition}
\begin{definition}
  For a subgraph $G'$ of $G$ and $W\subseteq \vertex{G'}$, we set 
  \begin{align*}
    \LPpdfx{\Phi}{G'}{\Vec{z}}{S}{T}{d_{G'}^W}&\defeq\Pr\left[\bigwedge_{\pi\in \Pi_{G'}}\sum_{uv\in \edge{\pi}}\Vec{a}[uv]\Vec{X}[uv]\!-\hspace*{-3mm}\sum_{\hspace*{2mm}v\in \vertex{\pi}}\hspace*{-2mm}d_{G'}^W(v) \le \Vec{z}[\src{\pi}]-\Vec{z}[\term{\pi}]\right].
  \end{align*}
\end{definition}
Since the vertex discount reduces the longest path length,
a greater vertex discount increases the probability.
That is, $\LPpdfx{\Phi}{G'}{\Vec{z}}{S}{T}{0}\le \LPpdfx{\Phi}{G'}{\Vec{z}}{S}{T}{d_{G'}^W}\le\LPpdfx{\Phi}{G'}{\Vec{z}}{S}{T}{d_{G'}^{W^+}}$ for any $v\in \vertex{G'}$ and $W\subseteq W^+\subseteq \vertex{G'}$,
where we write
$\LPpdfx{\Phi}{G'}{\Vec{z}}{S}{T}{0}=\LPpdf{\Phi}{G'}{\Vec{z}}{S}{T}$.
We consider discrete version of Proposition~\ref{proposition:switch} by
using $W$ of $d_{G'}^W$.
We merge $\LPpdfx{\Phi}{G_i}{\Vec{z}}{S}{T}{d_{G_i}^{\JS_i}}$ and
$\LPpdfx{\Phi}{U_i}{\Vec{z}}{S}{T}{d_{U_i}^{Q_i\cup\JT_i}}$ by the following. 
\begin{lemma}
  \label{lemma:horizontal_error_sum}
  Assume a separated tree decomposition with width $k$.
  Let $W_i=J_i\cup Q_i$ where $Q_i\subseteq \vertex{U_i}$.
  Then, we have
  \begin{align*}
    \LPpdfx{\Phi}{D_i}{\Vec{z}}{S}{T}{d_{D_i}^{W_i}}=\int_{\mathbb{R}^{J_i}} \left(\vecdiff{\JS_i}\LPpdfx{\Phi}{G_i}{\Vec{z}}{S}{T}{d_{G_i}^{\JS_i}}\right)\left(\vecdiff{\JT_i}\LPpdfx{\Phi}{U_i}{\Vec{z}}{S}{T}{d_{U_i}^{Q_i\cup\JT_i}}\right){\rm d}\Vec{z}[J_i].
  \end{align*}
\end{lemma}
\begin{proof}
We put the vertex discount into the graph
by considering an alternative construction of
the separated tree decomposition.
Given $G$ and its tree decomposition ${\cal T}(G)=({\cal B},{\cal A})$,
we construct $\hat{G}=(\hat{V},E_E\cup \hat{E_V})$ as follows.
We set $\hat{V}=\{v^-,v^+|v\in \vertex{G}\}\cup\{v^{(i)}|v\in B_i, B_i\in {\cal B}\}$,
$E_E=\{u^-v^+|uv\in\edge{G}\}$
and $\hat{E_V}=\{v^+v^{(i)},v^{(i)}v^-|v\in\vertex{G}\}$.
The length of edge $u^-v^+\in E_E$ for $uv\in \edge{G}$ is $X_{uv}$;
the length of edge $v^+v^{(i)}\in \hat{E_V}$ for $v\in \vertex{G}$
is static value $0$.
Then, we set the length of edge $v^{(i)}v^-\in \hat{E_V}$ as
static value $d_{G_i}^{\JS_i}(v)$.
We define bag-subgraph $\hat{G_i}$ of $\hat{G}$ by
$\vertex{\hat{G_i}}=\{v^+,v^{(i)},v^-|v\in B_i\}$ and $\edge{\hat{G_i}}=\{u^-v^+|uv\in\edge{G_i}\}\cup\{v^+v^{(i)},v^{(i)}v^-|v\in \vertex{G_i}\}$.
We define, for $B_i\in{\cal B}$, the subtree-subgraph $\hat{D_i}$ and
the uncapped subtree-subgraph $\hat{U_i}$ by replacing $G_i$'s by $\hat{G_i}$
in Definition~\ref{definition:subtreesubgraph}.
Now, we have $\LPpdfx{\Phi}{\hat{D_i}}{\Vec{z}}{S}{T}{0}=\LPpdfx{\Phi}{D_i}{\Vec{z}}{S}{T}{d_{D_i}^{W_i}}$.

By Lemma \ref{lemma:phi_increment},
the right hand side of the claim
is the longest path length distribution function
in $\hat{G_i} \cup \hat{U_i}$.
Notice that the longest path length in $\hat{G_i} \cup \hat{U_i}$
is equal to that of $\hat{D_i}$ since
the longest path length between $v^+$ and $v^-$ is equal to
$d_{D_i}^{W_i}(v)=\max\{d_{G_i}^{\JS_i}(v),d_{U_i}^{Q_i\cup\JT_i}(v)\}$,
implying the claim.
\end{proof}

We extend the other functions. 
$\LPpdfx{N}{G_i}{\Vec{z}}{S}{T}{M,d_{G_i}^W}$ is the number of cells
in $P_i=[0,1]^{\edge{G_i}}$
intersecting $\LPpdfx{K'}{G_i}{\Vec{z}}{S}{T}{d_{G_i}^W}$, where
\begin{align*}
  \LPpdfx{K'}{G_i}{\Vec{z}}{S}{T}{d_{G_i}^W}\!\!\defeq\!\!\left\{\!\Vec{x}\!\in\! P_i\!\left| \bigwedge_{\pi\in \Pi_{G_i}}\hspace*{-4mm}\sum_{\hspace*{4mm}uv\in \edge{\pi}}\hspace*{-3mm}\Vec{a}[uv]\Vec{x}[uv]\! -\hspace*{-3mm}\sum_{\hspace*{2mm}v\in \vertex{\pi}}\hspace*{-2mm}d_{G_i}^W(v)\!\le\!\Vec{z}[\src{\pi}]\!-\!\Vec{z}[\term{\pi}]\! \right.\right\}.
\end{align*}
In the inequalities of the conditions for
$\LPpdfx{K'}{G_i}{\Vec{z}}{S}{T}{d_{G_i}^W}$,
we have $\sum_{v\in \vertex{\pi}}d_{G_i}^W(v)$ instead of just $d_{G_i}^W(\src{\pi})$
since our definition allows multiple subgraph sources on a path.
Then, $\LPpdfx{\Lambda}{G_i}{\Vec{z}}{S}{T}{M,d_{G_i}^W}=\LPpdfx{N}{G_i}{x\Vec{g}/M}{S}{T}{M,d_{G_i}^W}M^{-|\edge{G_i}|}$
where $\Vec{g}[s]=\lceil M\Vec{z}[s]/x \rceil$ for $s\in S_i$ and
$\Vec{g}[t]=\lfloor M\Vec{z}[t]/x\rfloor$ for $t\in T_i$.

To bound the horizontal error,
we prove the following lemmas.
We write 
$\LPpdfx{\Lambda}{G_i}{\Vec{z}}{S}{T}{M,0}=\LPpdfx{\Lambda}{G_i}{\Vec{z}}{S}{T}{M}$, and $d_{G_i}=d_{G_i}^\emptyset$.
\begin{lemma}
  \label{lemma:sandwitch1}
  Given a separated tree decomposition with width $k$, we have
  \begin{align*}
    \LPpdfx{\Phi}{G_i}{\Vec{z}}{S}{T}{0} \le \LPpdfx{\Lambda}{G_i}{\Vec{z}}{S}{T}{M,0} \le \LPpdfx{\Phi}{G_i}{\Vec{z}}{S}{T}{d_{G_i}}.
  \end{align*}
\end{lemma}
\begin{proof}
Since the earlier inequality is evident by definition,
we focus on the latter.
Let $C\subseteq P_i=[0,1]^{\edge{G_i}}$ 
be any cell intersecting
$\LPpdfx{K'}{G_i}{\Vec{z}}{S}{T}{0}=\LPpdf{K'}{G_i}{\Vec{z}}{S}{T}$.
The latter inequality
$\LPpdfx{\Lambda}{G_i}{\Vec{z}}{S}{T}{M,0} \le \LPpdfx{\Phi}{G_i}{\Vec{z}}{S}{T}{d_{G_i}}$ is clear if we have
$C\subseteq \LPpdfx{K'}{G_i}{\Vec{z}}{S}{T}{d_{G_i}}$ for all $C$
such that $C\cap \LPpdfx{K'}{G_i}{\Vec{z}}{S}{T}{0}\neq \emptyset$. 

To prove $C\subseteq \LPpdfx{K'}{G_i}{\Vec{z}}{S}{T}{d_{G_i}}$, 
let $\Vec{p}, \Vec{q} \in C$ be two extreme points of $C$
so that 
$\Vec{p}=\frac{1}{M}\Vec{g}_E\in \LPpdfx{K'}{G_i}{\Vec{z}}{S}{T}{0}$ and $\Vec{q}=\frac{1}{M}(\Vec{g_E}+\Vec{1})\not\in \LPpdfx{K'}{G_i}{\Vec{z}}{S}{T}{0}$,
where $\Vec{g}_E\in \{0,1,\dots,M-1\}^{\edge{G_i}}$.
Remember that each component of
a point in $C\subseteq [0,1]^{\edge{G_i}}$ gives the
edge length of an edge in $G_i$.
gives the maximum difference of all path lengths in $G_i$.
Since $\Vec{p}\in\LPpdfx{K'}{G_i}{\Vec{z}}{S}{T}{0}$ and
$\Vec{q}\not\in\LPpdfx{K'}{G_i}{\Vec{z}}{S}{T}{0}$, 
for any fixed $\Vec{z}[S_i\cup T_i]$,
the edge lengths given by $\Vec{q}$
let some path $\pi$ violate the condition
in the definition of $\LPpdfx{K'}{G_i}{\Vec{z}}{S}{T}{d_{G_i}}$ 
so that
\begin{align*}
  \sum_{e\in \edge{\pi}}\Vec{a}[e]\Vec{p}[e]  &\le \Vec{z}[\src{\pi}]-\Vec{z}[\term{\pi}] \le \sum_{e\in \edge{\pi}} \Vec{a}[e]\Vec{q}[e]= \sum_{e\in \edge{\pi}} \Vec{a}[e]\left(\Vec{p}[e]+\frac{1}{M} \right)\\
  &\le \Vec{z}[\src{\pi}]-\Vec{z}[\term{\pi}]+\sum_{e\in\edge{\pi}}\frac{\Vec{a}[e]}{M}
  \le\Vec{z}[\src{\pi}]-\Vec{z}[\term{\pi}]+d_{G_i}(\src{\pi}).
\end{align*}
The last inequality is due to the assumption $\Vec{a}[e]\le x$
for all $e\in\edge{G}$ (see Proposition~\ref{proposition:atmostx}).
Thus, $\Vec{q}\in \LPpdfx{K'}{G_i}{\Vec{z}}{S}{T}{d_{G_i}}$,
implying
$C\subseteq \LPpdfx{K'}{G_i}{\Vec{z}}{S}{T}{d_{G_i}}$. 
\end{proof}

Let $W_i=J_i\cup Q_i$, where $Q_i\subseteq \vertex{U_i}$. For leaf bag $B_c\in {\cal B}$, we set
$\LPpdfx{\Lambda}{D_c}{\Vec{z}}{S}{T}{M,d_{D_c}^{J_c}}=\LPpdfx{\Lambda}{G_c}{\Vec{z}}{S}{T}{M,d_{G_c}^{\JS_c}}$.
Then, similarly as (\ref{form:Lambda_i}), we set
\begin{align*}
  \LPpdfx{\Lambda}{D_i}{\Vec{z}}{S}{T}{M,d_{D_i}^{W_i}}\!\defeq\hspace*{-4mm}\sum_{\Vec{g}\in\{0,\dots,M\}}^{J_i}\hspace*{-4mm}\Delta(\Vec{z}[\JS_i])\LPpdfx{\Lambda}{G_i}{\Vec{z}}{S}{T}{M,d_{G_i}^{\JS_i}}\Delta(\Vec{z}[\JT_i])\LPpdfx{\Lambda}{U_i}{\Vec{z}}{S}{T}{M,d_{U_i}^{Q_i\cup\JT_i}}. 
\end{align*}

In the following,
we consider the derivative of
$\LPpdfx{\Lambda}{D_i}{\Vec{z}}{S}{T}{M,d_{G_i}^{W}}$,
a staircase function, with respect to
$\Vec{z}$ by using the Dirac's delta function $\delta(x)$.
\begin{lemma}
\label{lemma:horizontal_error2}
Given a separated tree decomposition with width $k$, for $B_i\in{\cal B}$, let
$\displaystyle\LPpdfx{\Gamma}{D_i}{\Vec{z}}{S}{T}{M,d_{D_i}^{W_i}}\defeq \int_{\mathbb{R}^{J_i}}\left(\vecdiff{\JS_i}\LPpdfx{\Phi}{G_i}{\Vec{z}}{S}{T}{d_{G_i}^{\JS_i}}\right)\left(\vecdiff{\JT_i}\LPpdfx{\Lambda}{U_i}{\Vec{z}}{S}{T}{M,d_{U_i}^{Q_i\cup\JT_i}}\right){\rm d}\Vec{z}[J_i]$,
where $W_i=J_i\cup Q_i$ and $Q_i\subseteq \vertex{U_i}$.
Then, $\LPpdfx{\Lambda}{D_i}{\Vec{z}}{S}{T}{M,d_{D_i}^{Q_i}}\le \LPpdfx{\Gamma}{D_i}{\Vec{z}}{S}{T}{M,d_{D_i}^{W_i}}$.

\end{lemma}
\begin{proof}
  By Proposition~\ref{proposition:switch} and the property of the delta function $\delta(x)$, we have
  \begin{align*}
    \LPpdfx{\Gamma}{D_i}{\Vec{z}}{S}{T}{M,d_{D_i}^{W_i}}&=\int_{\mathbb{R}^{J_i}}\LPpdfx{\Phi}{G_i}{\Vec{z}}{S}{T}{d_{G_i}^{\JS_i}}(-1)^{|\JS_i|}\vecdiff{J_i}\LPpdfx{\Lambda}{U_i}{(\Vec{z})}{S}{T}{M,d_{U_i}^{Q_i\cup\JT_i}}{\rm d}\Vec{z}[J_i]\\
    &=\hspace*{-5mm}\sum_{\Vec{g}\in \{0,\dots,M\}^{J_i}}\hspace*{-5mm}\LPpdfx{\Phi}{G_i}{\Vec{u}}{S}{T}{d_{G_i}^{\JS_i}}(-1)^{|\JS_i|}\Delta(\Vec{z}[J_i])\LPpdfx{\Lambda}{U_i}{\Vec{u}}{S}{T}{M,d_{U_i}^{Q_i\cup\JT_i}}, 
  \end{align*}
  where $\Vec{u}[v]=\frac{x}{M}\Vec{g}[v]$ for $v\in J_i$ and
  $\Vec{u}[v]=\Vec{z}[v]$ for $v\not\in J_i$.

  Consider the left hand side of the claim.
  We prove the discrete analogy of Proposition~\ref{proposition:switch}
  to transform the definition of
  $\LPpdfx{\Lambda}{D_i}{\Vec{z}}{S}{T}{M,0}$.
  For $f(x)$ and $g(x)$,
  consider ``summation by parts''~\cite{GKP}.
  Let $\Delta f(x)=f(x+1)-f(x)$.
  By transforming $\Delta(f(x)g(x))=f(x+1)g(x+1)-f(x)g(x)$ 
  into $f(x)\Delta g(x)=\Delta(f(x)g(x))-f(x+1)\Delta g(x)$,
  the sum for $x=0,\dots,M$ gives
  \begin{align*}
    \sum_{x=0}^{M}f(x)\Delta g(x)=f(M+1)g(M+1)-f(0)g(0)-\sum_{x=0}^{M}g(x+1)\Delta f(x).
  \end{align*}
  We apply this summation by parts to
  $\LPpdfx{\Lambda}{D_i}{\Vec{z}}{S}{T}{M,d_{D_i}^{Q_i}}$,
  where $x$ is replaced by $\Vec{u}[v]$ for $v\in J_i$.
  For $W\subseteq J_i$, let a vector $\Vec{\theta}_W$ satisfy
  $\Vec{\theta}_W[v]=\frac{x}{M}$ for $v\in W$ and
  $\Vec{\theta}_W[v]=0$ for $v\not\in W$.
  Notice that
  $\LPpdfx{\Lambda}{G_i}{\Vec{u}}{S}{T}{M,0}\LPpdfx{\Lambda}{U_i}{\Vec{u}}{S}{T}{M,d_{U_i}^{Q_i}}=0$ if $\Vec{g}[v]\le 0$ or $\Vec{g}[v]\ge M$
  for any $v\in J_i$.
  By Lemma~\ref{lemma:sandwitch1}, and the summation by parts,
  \begin{align*}
    &\LPpdfx{\Lambda}{D_i}{\Vec{z}}{S}{T}{M,d_{D_i}^{Q_i}}=\hspace*{-6mm}\sum_{\Vec{g}\in\{0,\dots,M\}^{J_i}} \hspace*{-6mm}\LPpdfx{\Lambda}{G_i}{(\Vec{u}+\Vec{\theta}_{\JS_i})}{S}{T}{M,0}(-1)^{|\JS_i|}\Delta(\Vec{z}[J_i])\LPpdfx{\Lambda}{U_i}{\Vec{u}}{S}{T}{M,d_{U_i}^{Q_i}}\\
    &\le\hspace*{-6mm}\sum_{\Vec{g}\in\{0,\dots,M\}^{J_i}} \hspace*{-6mm}\LPpdfx{\Phi}{G_i}{(\Vec{u}+\Vec{\theta}_{\JS_i})}{S}{T}{d_{G_i}}(-1)^{|\JS_i|}\Delta(\Vec{z}[J_i])\LPpdfx{\Lambda}{U_i}{\Vec{u}}{S}{T}{M,d_{U_i}^{Q_i}}\\
    &=\hspace*{-6mm}\sum_{\Vec{g}\in\{0,\dots,M\}^{J_i}} \hspace*{-6mm}(-1)^{|J_i|}\Delta(\Vec{z}[\JT_i])\left(\LPpdfx{\Phi}{G_i}{(\Vec{u}+\Vec{\theta}_{\JS_i})}{S}{T}{d_{G_i}}\right)\Delta(\Vec{z}[\JS_i])\LPpdfx{\Lambda}{U_i}{(\Vec{u}+\Vec{\theta}_{\JT_i})}{S}{T}{M,d_{U_i}^{Q_i}}\\
    &\le\hspace*{-6mm}\sum_{\Vec{g}\in\{0,\dots,M\}^{J_i}} \hspace*{-6mm}\LPpdfx{\Phi}{G_i}{(\Vec{u}+\Vec{\theta}_{J_i})}{S}{T}{d_{G_i}}(-1)^{|S_i|}\Delta(\Vec{z}[J_i])\LPpdfx{\Lambda}{U_i}{(\Vec{u}+\Vec{\theta}_{\JT_i})}{S}{T}{M,d_{U_i}^{Q_i}}.
  \end{align*}
  The last part is at most $\LPpdfx{\Gamma}{D_i}{\Vec{z}}{S}{T}{M,d_{D_i}^{W_i}}$ by $W_i=\JS_i\cup\JT_i\cup Q_i=J_i\cup Q_i$,
  since $\LPpdfx{\Phi}{G_i}{(\Vec{u}+\Vec{\theta}_{J_i})}{S}{T}{d_{G_i}}\le\LPpdfx{\Phi}{G_i}{(\Vec{u}+\Vec{\theta}_{\JS_i})}{S}{T}{d_{G_i}}$.
\end{proof}

The following lemma proves the total vertex discount for the entire graph.
\begin{lemma}
\label{lemma:sandwitch2}
Assume a width $k$ separated tree decomposition.
For the root bag $B_0\in {\cal B}$, we have
$\LPpdfx{\Phi}{D_0}{\Vec{z}}{S}{T}{0}\le \LPpdfx{\Lambda}{D_0}{\Vec{z}}{S}{T}{M,0} \le \LPpdfx{\Phi}{D_0}{\Vec{z}}{S}{T}{(A_0+1)d_{D_0}^{W_0}}$,
where $A_i$ for $B_i\in{\cal B}$ is the maximum distance from $B_i$ to any leaf
in ${\cal T}(G)$.
\end{lemma}
\begin{proof}
Since the inequality on the left is evident by definition,
we prove the inequality on the right.
We prove that 
$\LPpdfx{\Lambda}{D_i}{\Vec{z}}{S}{T}{M,0} \le \LPpdfx{\Phi}{D_i}{\Vec{z}}{S}{T}{(A_i+1)d_{D_i}^{W_i}}$ from the leaves to the root
by induction on $i$.

As for the base case,  
for any leaf bag $B_c\in {\cal B}$ where $A_c=0$ and $W_c=\emptyset$,
we have
$\LPpdfx{\Lambda}{D_c}{\Vec{z}}{S}{T}{M,0}=\LPpdfx{\Lambda}{G_c}{\Vec{z}}{S}{T}{M,0}\le \LPpdfx{\Phi}{D_c}{\Vec{z}}{S}{T}{d_{D_c}}$ 
by Lemma \ref{lemma:sandwitch1}.

We assume that
$\LPpdfx{\Lambda}{D_j}{\Vec{z}}{S}{T}{M,0}\le \LPpdfx{\Phi}{D_j}{\Vec{z}}{S}{T}{A_id_{D_j}^{W_j}}$ for
$j\in \Ch(i)$ as the induction hypothesis.
Consider the middle part $\LPpdfx{\Lambda}{D_i}{\Vec{z}}{S}{T}{M,0}$
of the claim, which is at most
$\LPpdfx{\Gamma}{D_i}{\Vec{z}}{S}{T}{M,d_{D_i}^{W_i}}$
by Lemma \ref{lemma:horizontal_error2}.
By Proposition~\ref{proposition:switch},
we have
\begin{align*}
  \LPpdfx{\Gamma}{D_i}{\Vec{z}}{S}{T}{M,d_{D_i}^{W_i}}=\int_{\mathbb{R}^{J_i}}(\vecdiff{J_i}\LPpdfx{\Phi}{G_i}{\Vec{z}}{S}{T}{d_{G_i}^{\JS_i}})~(-1)^{|\JT_i|}\LPpdfx{\Lambda}{U_i}{\Vec{z}}{S}{T}{M,d_{U_i}^{Q_i\cup\JT_i}}{\rm d}\Vec{z}[J_i].
\end{align*}
By the induction hypothesis, this is at most
\begin{align*}
  &\int_{\mathbb{R}^{J_i}} (\vecdiff{J_i}\LPpdfx{\Phi}{G_i}{\Vec{z}}{S}{T}{d_{G_i}^{\JS_i}})~(-1)^{|\JT_i|}\LPpdfx{\Phi}{U_i}{\Vec{z}}{S}{T}{A_id_{U_i}^{Q_i}+d_{U_i}^{Q_i\cup\JT_i}}{\rm d}\Vec{z}[J_i],
\end{align*}
which is at most $\LPpdfx{\Phi}{D_i}{\Vec{z}}{S}{T}{(A_i+1)d_{D_i}^{W_i}}$
due to Proposition~\ref{proposition:switch} and
Lemma \ref{lemma:horizontal_error_sum}, implying the lemma.
\end{proof}

Now, we transform the total vertex discount into the horizontal error.
Remember that $b=|{\cal B}|$ is the number of bags in ${\cal T}(G)$.
\begin{lemma}
  \label{lemma:Upperbound}
$\displaystyle \LPpdfx{\Phi}{D_0}{x\Vec{\sigma}}{S}{T}{d_{D_0}^{W_0}}\le \LPpdfx{\Phi}{D_0}{(x+(k+1)bnx/M)\Vec{\sigma}}{S}{T}{0}$
\end{lemma}
\begin{proof}
By definition, we have 
$d_{D_0}^{W_0}(v)\le \frac{(k+1)x}{M}$ for any $v\in V$.
Therefore, the sum of the vertex discount $(A_0+1)d_{D_0}^{W_0}$
in a path is at most $(k+1)bnx/M$.
We get the claim by adding this value to $x\Vec{\sigma}[\Src{G}]$. 
\end{proof}

Lemma~\ref{lemma:bound} shows the vertical approximation ratio,
which proves Theorem~\ref{th:runningtime} in combination with Observation~\ref{observation:runningtime}.
The arguments are similar to \cite{AK2016}.
\begin{lemma}
  \label{lemma:bound}
  Let $M=\lceil 2(k+1)bmn/\epsilon \rceil$.
  Given a width $k$ separated tree decomposition, we have
    $1\le \LPpdfx{\Lambda}{D_0}{x\Vec{\sigma}}{S}{T}{M,0}/\LPpdfx{\Phi}{D_0}{x\Vec{\sigma}}{S}{T}{0}\le 1+\epsilon$.
\end{lemma}
\begin{proof}
Since the earlier inequality is clear by definition, we prove the latter 
inequality.
Let $\eta=1+(k+1)bn/M$.
Since $\LPpdfx{\Lambda}{D_0}{x\Vec{\sigma}}{S}{T}{M,0}$ is at most
$\LPpdfx{\Phi}{D_0}{x\Vec{\sigma}}{S}{T}{d_{D_0}}$, which is at most
$\LPpdfx{\Phi}{D_0}{\eta x\Vec{\sigma}}{S}{T}{0}$
by Lemma \ref{lemma:sandwitch2} and \ref{lemma:Upperbound}, we have
$\LPpdfx{\Phi}{D_0}{x\Vec{\sigma}}{S}{T}{0}/\LPpdfx{\Lambda}{D_0}{x\Vec{\sigma}}{S}{T}{M,0}\ge \LPpdfx{\Phi}{D_0}{x\Vec{\sigma}}{S}{T}{0}/\LPpdfx{\Phi}{D_0}{\eta x\Vec{\sigma}}{S}{T}{0}$ by considering the reciprocal of the approximation ratio.

We claim that 
$\LPpdfx{\Phi}{D_0}{x\Vec{\sigma}}{S}{T}{0}/\LPpdfx{\Phi}{D_0}{\eta x\Vec{\sigma}}{S}{T}{0}\ge \eta^{-m}= \left(1+\frac{\epsilon}{2m}\right)^{-m}$, 
which we verify as follows.
By definition, $\LPpdfx{\Phi}{D_0}{x\Vec{\sigma}}{S}{T}{0}$ is the volume of $K_G(\Vec{a},x)=\LPpdfx{K'}{G}{x\Vec{\sigma}}{S}{T}{0}$.
Here, scaling $K_G(\Vec{a},x)$ by $\eta$
gives another polytope 
\begin{align*}
  \tilde{K_G}(\Vec{a},x)&\defeq \{\Vec{x} \in \mathbb{R}^{m} | \exists \Vec{y}\in K_G(\Vec{a},x)~~\text{s.t.}~~ \Vec{x}=\eta\Vec{y}\}.
\end{align*}
Then, it is clear that $K_G(\Vec{a},\eta x)\subseteq \tilde{K_G}(\Vec{a},x)$.
Thus,
\begin{align*}
&\frac{\LPpdfx{\Phi}{D_0}{x\Vec{\sigma}}{S}{T}{0}}{\LPpdfx{\Phi}{D_0}{\eta x\Vec{\sigma}}{S}{T}{0}}
  = \frac{{\rm Vol}(K_G(\Vec{a},x))}{{\rm Vol}(K_G(\Vec{a},\eta x))} \ge \frac{{\rm Vol}(K_G(\Vec{a},x))}{{\rm Vol}(\tilde{K_G}(\Vec{a},x))}.
\end{align*}
The rightmost hand side is equal to
$\eta^{-m} \ge \left(1+\frac{\epsilon}{2m}\right)^{-m}$.
Now, we have 
\begin{align*}
  \left(1+\frac{\epsilon}{2m}\right)^{-m}\ge \left(1-\frac{\epsilon}{2m}\right)^{m}\ge 1-\frac{m\epsilon}{2m}=1-\frac{\epsilon}{2}.
\end{align*}
The first inequality is because $((1+\frac{\epsilon}{2m})(1-\frac{\epsilon}{2m}))^m\le 1$.
Then, we have, for $0\le \epsilon\le 1$,
$\LPpdfx{\Phi}{D_0}{\eta x\Vec{\sigma}}{S}{T}{0}/\LPpdfx{\Phi}{D_0}{x\Vec{\sigma}}{S}{T}{0} \le 1/(1-\epsilon/2)\le 1+\epsilon$,
implying the lemma. 
The restriction $\epsilon \le 1$ is not essential. If $\epsilon>1$, we replace $\epsilon$ by $1$.
\end{proof}

\section{Other Edge Length Distributions}
The edge lengths that follow another distribution give
us a different problem.
Assuming well-behaved 
distribution and constant treewidth graph,
we can compute $\Pr[X_{\rm MAX}\le x]$ exactly or approximately.
After some preliminaries,
we present two results. The first result is about the case
the edge lengths follow the standard exponential distribution.
The second result is that 
the edge lengths follow some abstract distributions
to which the Taylor approximation is applicable.
As in the previous section, we first
assume a separated tree decomposition with width $k$.
Then, we replace the treewidth $k$ with $2k+1$ to construct the
separated tree decomposition. 

\subsection{Integration of Step Functions}
We need some preliminaries about
the integrals using step functions.
\begin{definition}
  Let $F(\Vec{x})$ and $\tilde F(\Vec{x})$ be two piecewise continuous functions
  for $\Vec{x}\in \mathbb{R}^k$.
  If $F(\Vec{x})=\tilde F(\Vec{x})$ for $\Vec{x}\in \mathbb{R}^k\setminus I$
  where
  the $k$-dimensional measure of $I$ is 0, then we write
  $F(\Vec{x})\AEEq \tilde F(\Vec{x})$.
  If $F(\Vec{x})\AEEq \tilde F(\Vec{x})$, we say $F(\Vec{x})$ and $\tilde F(\Vec{x})$
  are {\em equal almost everywhere}.
\end{definition}
We use step function $H(x)$ for describing the function with cases.
In case $F(x)=1$ for any $x\in \mathbb{R}$,
we have $F(x)\AEEq H(x)+H(-x)$ where the equality does not hold for $x=0$.
Though we may have some wrong values at the breakpoints,
they do not matter 
when executing another integral with $F(x)$ as a factor of the integrand.
In the following, we use $H(x)$ in describing the functions with cases.

Executing an integral of a step function
needs some attention. 
Let $f(x)$ be a function of $x\in \mathbb{R}$ such that
$\lim_{x\rightarrow -\infty} f(x)=0$
and $F(x)=\int_{-\infty}^{x}f(x){\rm d}x$ converges to a finite value.
For any $\Vec{y},\Vec{z}\in \mathbb{R}^{n}$, we have
\begin{align*}
  \int_{\mathbb{R}}\left(\prod_{i=1,\dots,k}\!\!\!H(x-\Vec{y}[i])\!\!\!\prod_{i=1,\dots,k}\!\!\!H(\Vec{z}[i]-x)\right)\!f(x){\rm d}x=H(z-y)(F(z)-F(y)),
\end{align*}
where $z=\min_{1\le i \le k}\{\Vec{z}[i]\}$, $y=\max_{1\le i \le k}\{\Vec{y}[i]\}$.
For convenience, we replace
$H(z-y)$, $F(z)$, and $F(y)$
by a sum of $k!$ terms of step function products.
That is, $F(z)=P(\Vec{z},1)$ and $F(y)=P(\Vec{y},k)$ using 
$P(\Vec{x},i)\AEEq \sum_{\Vec{p} \in \mathrm{Perm}(k)}F(\Vec{x}[p_i])\prod_{1\le j \le k-1}H(\Vec{x}[p_{j+1}]-\Vec{x}[p_j])$.
Here, $\mathrm{Perm}(k)$ is the set of all $k!$ permutations of
${1,\dots,k}$.
Each element $\Vec{p}$ of $\mathrm{Perm}(k)$ is 
a $k$-tuple $\Vec{p}=(p_1,\dots,p_k)$.
Though the right hand side
may have wrong values at points where
any two variables are equal (e.g., $\Vec{z}[p_j]=\Vec{z}[p_{j+1}]$ or
$\Vec{y}[p_j]=\Vec{y}[p_{j+1}]$),
we can ignore the wrong values.
If $F(\Vec{z})$ is well-behaved like
exponential distribution function, we recover
the values $F(\Vec{z})$ at any point $\Vec{z}$
where $\Vec{z}[p_i]=\Vec{z}[p_{i+1}]$.
That is, we may compute the value of 
$\lim_{t\rightarrow +0} F(\Vec{z}+t\Vec{h})$,
as long as we have an expression of $F(\Vec{z})$
satisfying that $F(\Vec{z})$ is
exact in a continuous region including $\Vec{z}+t\Vec{h}$
and that there is no breakpoint of
step function factors between $\Vec{z}$ and $\Vec{z}+\Vec{h}$.

\subsection{Exact Computation for Exponentially Distributed Edge Lengths}
Here, we consider the case where the edge lengths are mutually
independent and follow the standard exponential distribution.
By the direct computation of the integrals, we obtain 
$\Pr[X_{\rm MAX}\le x]$ exactly.
The following is the definition of the standard exponential distribution.
\begin{definition}
  If random edge length $X$ satisfies
  $\Pr[X\le x]=H(x)(1-e^{-x})$,
  we say $X$ follows the {\em standard exponential distribution}.
\end{definition}
Assuming that we can compute Napier's constant $e$
and its power in $O(1)$ time, we have the following theorem.
\begin{theorem}
\label{th:exponential_distribution}
Let $G=(V,E)$ be a DAG with treewidth at most $k$.
Assume that the edge lengths are mutually independent and
follow the standard exponential distribution.
There is an algorithm that exactly computes
$\Pr[X_{\rm MAX}\le x]$ in $((4k+2)mn)^{O(k)}$ time.
\end{theorem}

Our algorithm is the following. We first
compute $\LPpdf{\phi}{G_i}{\Vec{z}}{S}{T}$ exactly for each bag $B_i\in {\cal B}$
by Theorem \ref{th:multiplesourceterminal}.
Then, we merge $\LPpdf{\phi}{G_i}{\Vec{z}}{S}{T}$'s starting
from the leaves to the root so that we have
$\LPpdf{\phi}{D_i}{\Vec{z}}{S}{T}$ for each $B_i$
by Lemma \ref{lemma:phi_increment}.
In the ongoing computation, we transform the integrand 
into a sum of terms where each term
is a product of the powers of $z_i$'s
and the exponential functions of $z_i$'s.
Thus, we can store the resulting form of the ongoing computation
in a $2k$-dimensional array.

To save some memory space,
we consider $\Vec{w}$ defined as follows.
Let $B_h$ be the parent of $B_i$.
In executing Step 08,
we put $\Vec{w}[v]=\Vec{z}[v]$ for $v\in (\Src{D_i}\cup \Term{D_i})\cap B_h$,
$\Vec{w}[s]=x$ for $s\in \Src{D_i}\setminus B_h$, and
$\Vec{w}[t]=0$ for $t\in \Term{D_i}\setminus B_h$.

The following is the algorithm for computing the value of
$\Pr[X_{\rm MAX}\le x]$.
\begin{algorithm}
  Exact-Exponential($G,{\cal T}(G), x$):\\
  Input: DAG $G$, binary tree decomposition ${\cal T}(G)$ of $G$ with bags $B_0,\dots,B_{b-1}$,\\
  \hspace*{1cm} and longest path length $x\in \mathbb{R}$;\\
  Output: Value of $\Pr[X_{\rm MAX}\le x]$;\\
  01. For each $B_i\in {\cal B}$ do:\\
  02. \hspace*{2mm} Compute $\LPpdf{\phi}{G_i}{\Vec{z}}{S}{T}$ by Theorem \ref{th:multiplesourceterminal};\\
  03. done;\\
  04. For each leaf bag $B_c\in {\cal B}$ of ${\cal T}(G)$ do:\\
  05. \hspace*{2mm} Set $\LPpdf{\phi}{D_c}{\Vec{w}}{S}{T}:=\LPpdf{\phi}{G_c}{\Vec{w}}{S}{T}$;\\
  06. done;\\
  07. For each bag $B_i\in {\cal B}$ from the leaves to the root $B_0$ do:\\
  08. \hspace*{2mm} Compute $\LPpdf{\phi}{D_i}{\Vec{w}}{S}{T}$ by (\ref{form:phiDi}) of Lemma~\ref{lemma:phi_increment};\\
  09. done;\\
  10. Compute and Output $\LPpdf{\Phi}{D_0}{x\Vec{\sigma}}{S}{T}$.
\end{algorithm}

The proof of Theorem \ref{th:exponential_distribution} is by 
direct estimation of the amount of space and time for storing and
computing $\LPpdf{\phi}{G_i}{\Vec{z}}{S}{T}$ and $\LPpdf{\phi}{D_i}{\Vec{z}}{S}{T}$.
Since we consider the time and space for generating
and storing the ongoing computation in the proof,
we first briefly describe what we consider is a form.
\begin{definition}
\label{definition:form}
We consider the following as a unit of 
a form: a real number, a variable and its power, a step function,
a Dirac's delta function, and an exponential function.
What we call a form is a unit in the above or 
a combination of these units by arithmetic operations:
addition, subtraction, multiplication, and division.
Each arithmetic operation has 
a unit or a parenthesized form on its sides.
The sum of products is a form where the parenthesized forms in it
consist of a product of some units.
A term in a form is a group of units connected continuously
by multiplications or divisions.
A coefficient of a term is a number that is obtained
by executing all multiplications 
and the division of numbers in the term.
\end{definition}

\begin{proof} (of Theorem \ref{th:exponential_distribution})
To estimate the algorithm's running time,
we first bound the number of terms
that appear in the form of $\LPpdf{\phi}{G_i}{\Vec{z}}{S}{T}$.
In a large part the proof, we assume
a separated tree decomposition with width $k$.
Then, the number of terms is bounded by using
the number of variables, which is at most $k$.
At the end of the proof, we replace the treewidth with
$2k+1$ to construct a separated tree decomposition.
We note that
\begin{align*}
  \LPpdf{\phi}{G_i}{\Vec{z}}{S}{T}\!=\!\!\int_{\mathbb{R}^{\internal{G_i}}}\prod_{u\in \internal{G_i}}\frac{\partial}{\partial z_u}\prod_{v\in\Suc(u)}\hspace*{-3mm}H(z_u-z_v)(1-e^{-(z_u-z_v)}){\rm d}\Vec{z}[\internal{G_i}].
\end{align*}
Since the standard exponential density function is $H(x)e^{-x}$,
all formulas that appear in the ongoing computation 
are the combinations of the step function $H(x)$,
the polynomial, 
and the exponential function.

Let us consider the step function factors in the integral.
Though one integral using step functions can produce many
cases in the resulting form, we break down
each of these cases as a sum of step function products.
Since the argument of each step function is a difference of
two dummy variables like $z_u-z_v$,
we can cover all cases by considering
all possible permutations of $z_v$'s for $v\in B_i$.
Thus, combinations of step function factors are
at most $(k+1)!$ for $\LPpdf{\phi}{G_i}{\Vec{z}}{S}{T}$.

We show that at most $n^{k+1}(2m+1)^{k+1}$
polynomial-exponential product terms may 
appear in $\LPpdf{\phi}{G_i}{\Vec{z}}{S}{T}$ for each permutation of
$z_v$'s for $v\in B_i$.
Observe that each term is a product of $z_v^{\alpha_v}$ and
$\exp(\beta_vz_v)$, where $\alpha_v$'s and $\beta_v$'s are integers.
We expand the integrand in (\ref{form:phiDi}) of
Lemma \ref{lemma:phi_increment}
into the sum of products.
Then, the degree $\beta_v$ of 
$\exp(z_v)$ can either increase or decrease at most by one
in taking a product of two distribution functions, and hence
$\beta_v$'s are integers between $-k^2$ and $k^2$.
Consider executing the indefinite integral
$\int z_v^{\alpha_v}e^{\beta_vz_v}{\rm d}z_v$.
In case $\beta_v=0$, $\int z_v^{\alpha_v}{\rm d}z_v=z_v^{\alpha_v+1}/(\alpha_v+1)+c$, where we may set $c=0$ here.
In case $\beta_v\neq 0$, integration by parts leads to
$\int z_v^{\alpha_v}e^{\beta_vz_v}{\rm d}z_v=\sum_{i=0}^{\alpha_v}(-1)^i(\alpha_v!/(\alpha_v-i)!)\beta_v^{-i-1}z_v^{\alpha_v-i}e^{\beta_v z_v}$.
The degree $\alpha_v$ of $z_v$ that appears 
in the terms in the formula of $\LPpdf{\phi}{G_i}{\Vec{z}}{S}{T}$
can increase at most by one in one integral
(in case $\beta_v=0$), and hence $\alpha_v$'s are
at most $k$ non-negative integers.
Now, in expanding the resulting form into the sum of products,
we have at most $(k+1)^k(2k^2+1)^k$ terms
for each permutation of the at most $k+1$ argument variables
in the representation of $\LPpdf{\phi}{G_i}{\Vec{z}}{S}{T}$.
Therefore, an array of
$O((k+1)!(k+1)^k(2k^2+1)^k)\le k^{O(k)}$ real numbers
is sufficient for the coefficients of
the terms in $\LPpdf{\phi}{G_i}{\Vec{z}}{S}{T})$.

By a similar argument, we bound the number of array elements
to store $\LPpdf{\phi}{D_i}{\Vec{z}}{S}{T}$.
Since the subtree graph is connected to the rest of the other parts
only by the vertices at the root of the subtree,
we have at most $k+1$ variables ($z_v$ for $v \in S_i$)
in $\LPpdf{\phi}{D_i}{\Vec{z}}{S}{T}$.
Similarly, as $\LPpdf{\phi}{G_i}{\Vec{z}}{S}{T}$,
we expand $\LPpdf{\phi}{D_i}{\Vec{z}}{S}{T}$ into 
the sum of products of step functions and polynomials and exponential
functions.
Here, $\LPpdf{\phi}{D_i}{\Vec{z}}{S}{T}$ is different from 
$\LPpdf{\phi}{G_i}{\Vec{z}}{S}{T}$ in that 
$0\le \alpha_v\le |\vertex{D_i}|$ and $|\beta_v|\le |\edge{D_i}|$.
Therefore, we can store $\LPpdf{\phi}{D_i}{\Vec{z}}{S}{T}$
in an array of real numbers with
$(k+1)! n^k(2m+1)^k\le (k+1)! (2mn)^{O(k)}$ elements.

Then, we bound the time to compute 
$\LPpdf{\phi}{D_i}{\Vec{z}}{S}{T}$ from $\LPpdf{\phi}{G_i}{\Vec{z}}{S}{T}$ and
$\LPpdf{\phi}{D_j}{\Vec{z}}{S}{T}$'s for $j\in \Ch(i)$.
Expanding the integrand into the sum of products
takes running time proportional to
the array size for $\LPpdf{\phi}{G_i}{\Vec{z}}{S}{T}$ multiplied
by the array size for 
$\LPpdf{\phi}{D_j}{\Vec{z}}{S}{T}~(j\in \Ch(i))$.
Since $|\Ch(i)|\le 2$ by our assumption,
the running time
of expanding the integrand is at most $k^{O(k)} (2mn)^{O(k)}$.
Since we can compute the antiderivative of $z_v^{\alpha}\exp(\beta z_v)$
in $O(\alpha)$ time for positive integer $\alpha \le n$ 
using integration by parts,
the time to integrate each term in $\LPpdf{\phi}{D_i}{\Vec{z}}{S}{T}$ is 
$k^{O(k)}O(n)$.
We execute the integration by
computing the integral of every possible term.
The running time for computing $\LPpdf{\phi}{D_i}{\Vec{z}}{S}{T}$
from $\LPpdf{\phi}{G_i}{\Vec{z}}{S}{T}$ and
$\LPpdf{\phi}{D_j}{\Vec{z}}{S}{T}$'s for $j\in \Ch(i)$ is
$k^{O(k)}O(n) (2mn)^{O(k)}\le (2kmn)^{O(k)}$.
Now, we repeat this computation for all $B_i\in {\cal B}$.

We obtain the running time in the claim by replacing the treewidth $k$
by $2k+1$ to construct a separated tree decomposition.
In this construction, though a delta function appears for each
vertex $v\in \vertex{G}$, the total running time to process
a delta function for each $v$
is linear to the length of an ongoing computation,
which does not change the $O$-notation.
To verify this, remember that
$v^+$ has only one outgoing edge $v^+v^-$ with static length $0$.
Then, we have $\delta(\Vec{z}[v^+]-\Vec{z}[v^-])$ for edge $v^+v^-$.
Since $v^+$ has exactly one outgoing edge,
we can always execute the integral with
respect to $\Vec{z}[v^+]$ ending up with that we replace $\Vec{z}[v^+]$
by $\Vec{z}[v^-]$ in the integrand.
\end{proof}

\subsection{Taylor Approximation of the Edge Length Distribution Function}
If we can use the Taylor approximation, 
we have an approximation scheme for $\Pr[X_{\rm MAX}\le x]$
with {\em additive} error $\epsilon'$.
Now we define the distribution function $F_{ij}(x)$ implicitly.
As a part of the input, we assume that we have an oracle that computes
$F_{uv}(x)$ for each edge $uv\in E$.
Let $\tau$ 
be a parameter of our algorithm.
We also assume that the oracle computes the exact value of
$F_{uv}^{(d)}(t)$ for $d\in {0,1,\dots,\tau}$ 
and $t\in [0,x]$ in $O(1)$ time, where $F_{uv}^{(d)}(x)$ is
the derivative of order $d$ of $F_{uv}(x)$.
For $F_{uv}(x)=\Pr[X_{uv}\le x]$, we assume following three conditions:
\begin{itemize}
\item The length $X_{uv}$ of each edge $uv\in E$ is nonnegative;
\item For any edge $uv\in E$,
  the distribution function $F_{uv}(x)$ satisfies that $|F^{(d)}_{uv}(t)|\le 1$ for any $d\in\{0,1,\dots,\tau\}$ 
  and $t\in [0,x]$;
\item The Taylor series of $F_{uv}(x)$ converges to $F_{uv}(x)$ itself.
\end{itemize}
In addition to Definition~\ref{definition:form}, the edge length distribution
function and its derivatives are also the units of the form.

As its input, our approximation algorithm takes $\epsilon'~~(0<\epsilon'\le 1)$ and $x>0$ as well as the DAG $G$ and its tree decomposition ${\cal T}(G)$.
Our algorithm outputs a polynomial $\tilde F_{\rm MAX}(x)$
that approximates $\Pr[X_{\rm MAX}\le x]$.
The constraint $\epsilon'\le 1$ is because $\Pr[X_{\rm MAX}\le x]\le 1$
by definition.
We have the following theorem.
\begin{theorem}
  \label{th:TaylorTreeWidth}
  Let $G$ be a DAG whose underlying undirected graph has
  treewidth at most $k$.
  Given the binary tree decomposition ${\cal T}(G)$ with $b=O(n)$ bags and
  the distribution function $F_{uv}(x)$ of the lengths of 
  each edge $uv\in E$ described as above.
  There exists an algorithm that computes $\tilde F_{\rm MAX}(x)$
  in time $O(n(2k+1)^{O(k)}\tau^{O(k^2)})$ satisfying
  $|\tilde F_{\rm MAX}(x) -\Pr[X_{\rm MAX}\le x]|\le \epsilon'$,
  where
  $\tau=\lceil(e^2+1)(2k+2)x+ 2\ln b +\ln (1/\epsilon')\rceil+1$.
\end{theorem}
Since it is easy to see that $F_{\rm MAX}(x)=0$ for $x\le 0$,
we concentrate on the case $x \ge 0$.
By the {\em Taylor approximation of} $f(\Vec{z})$, 
we mean the Taylor approximation generated by 
$f(\Vec{z})$ at the origin $\Vec{z}=\Vec{0}$.
The following is what we call the Taylor approximation (e.g., \cite{CJ1989}).
\begin{definition}
  Let $\Vec{z}=(z_1,\dots,z_k)$.
  Let $F(\Vec{z})$ be a $k$-variable function.
  The order $\tau$ 
  Taylor approximation of $F(\Vec{z})$ is
  $\tilde F(\Vec{z})\defeq \sum_{i=0}^\tau \frac{1}{i!}\left(\sum_{j=1}^k z_j\frac{\partial}{\partial z_j}\right)^iF(\Vec{0})$.
\end{definition}
Here, $\tilde F(\Vec{z})$ approximates $F(\Vec{z})$ well for
sufficiently large (but not very large) $\tau$. 
It is well known that 
the additive error of the approximation has the following bound.
\begin{proposition}
  \label{proposition:Taylor_error}
  Let $\Vec{z}=(z_1,\dots,z_k)\in \mathbb{R}^k$.
  For a $k$-variable function $F(\Vec{z})$
  satisfying
  $\left|\left(\sum_{j=1}^k\frac{\partial}{\partial z_j}\right)^{\tau+1}F(\Vec{z})\right|\le k^{\tau+1}$,
  for any $\Vec{0}\le \Vec{z}\le x\Vec{1}$,
  we have
  $|\tilde F(\Vec{z})-F(\Vec{z})| \le \frac{(kx)^{\tau+1}}{(\tau+1)!}$.
\end{proposition}
The proof is a straightforward extension
of the proof for the two variable functions' remainder in \cite{CJ1989}.

Our algorithm is as follows.
For each bag $B_i\in {\cal B}$, we compute the order $\tau$
Taylor approximation $\LPpdf{\tilde \phi}{G_i}{\Vec{z}}{S}{T}$ of $\LPpdf{\phi}{G_i}{\Vec{z}}{S}{T}$.
In computing $\LPpdf{\tilde \phi}{G_i}{\Vec{z}}{S}{T}$, we compute
the order $\tau$ 
Taylor approximation of the integrand of
the form in Theorem \ref{th:multiplesourceterminal}.
Remember the assumption that
we can compute the integrand and its derivatives exactly.
Then, as $\LPpdf{\tilde \phi}{G_i}{\Vec{z}}{S}{T}$,
we compute the order $\tau$ 
Taylor approximation of the
resulting form after executing all integrals.
For each leaf bags $B_c \in {\cal B}$, we set
$\LPpdf{\tilde \phi}{D_c}{\Vec{z}}{S}{T}\defeq \LPpdf{\tilde \phi}{G_c}{\Vec{z}}{S}{T}$.
Let $\LPpdf{\tilde \Phi}{U_i}{\Vec{z}}{S}{T}=\prod_{j\in\Ch(i)}\LPpdf{\tilde \Phi}{D_j}{\Vec{z}}{S}{T}$ and $\LPpdf{\tilde \phi}{U_i}{\Vec{z}}{S}{T}=\vecdiff{\Src{U_i}}\LPpdf{\tilde \Phi}{U_i}{\Vec{z}}{S}{T}$.
Here, $\LPpdf{\tilde \phi}{D_i}{\Vec{z}}{S}{T}$ is an approximation of
\begin{align}
  \LPpdf{\psi}{D_i}{\Vec{z}}{S}{T}\defeq\int_{(0,x\Vec{1}[J_i])} \LPpdf{\tilde \phi}{G_i}{\Vec{z}}{S}{T}\LPpdf{\tilde \phi}{U_i}{\Vec{z}}{S}{T} {\rm d}\Vec{z}[J_i], \label{form:tilde_phi_increment}
\end{align}
which is analogous to Lemma \ref{lemma:phi_increment}.
In approximating (\ref{form:tilde_phi_increment}),
every time we execute an integral with respect to $z_v$
for $v\in J_i$, we compute the order $\tau$ 
Taylor approximation of the resulting form.
Then, we have $\LPpdf{\tilde \phi}{D_i}{\Vec{z}}{S}{T}$ as the order $\tau$ 
Taylor approximation
of the resulting form
after executing all integrals in (\ref{form:tilde_phi_increment}).
Now, we put
$\LPpdf{\tilde\Phi}{D_i}{\Vec{z}}{S}{T}\defeq\int_{(0,\Vec{z}[\Src{D_i}])}\LPpdf{\tilde\phi}{D_i}{\Vec{z}}{S}{T}{\rm d}\Vec{z}[\Src{D_i}]$.

To save the memory space, we define $\Vec{w}$ as follows.
Let $B_h$ be the parent of $B_i\in {\cal B}$.
We put $\Vec{w}[v]=\Vec{z}[v]$ for $v\in (\Src{D_i}\cup \Term{D_i})\cap B_h$,
$\Vec{w}[s]=x$ for $s\in \Src{D_i}\setminus B_h$ and
$\Vec{w}[t]=0$ for $t\in \Term{D_i}\setminus B_h$.
The following is our algorithm.
\begin{algorithm}
  Approx-Taylor($G, {\cal T}(G), x,\epsilon'$)\\
  Input: DAG $G$,  
  binary tree decomposition ${\cal T}(G)$,
  longest path length $x\in \mathbb{R}$, and
  the additive error $\epsilon'$; \\
  Output: Approximate value of $\Pr[X_{\rm MAX}\le x]$;\\
  01. Set $\tau := (e^2+1)(k+1)x+ 2\ln b +\ln (1/\epsilon')$;\\ 
  02. For each $B_i\in {\cal B}$ do \\
  03. \hspace*{2mm} Compute $\LPpdf{\tilde \phi}{G_i}{\Vec{z}}{S}{T}$ as the order $\tau$ 
  Taylor approximation of $\LPpdf{\phi}{G_i}{\Vec{z}}{S}{T}$;\\
  04. done; \\
  05. For each leaf bag $B_c\in {\cal B}$ of ${\cal T}(G)$ do:\\
  06. \hspace*{2mm} Set $\LPpdf{\tilde \phi}{D_c}{\Vec{z}}{S}{T}:=\LPpdf{\tilde \phi}{G_c}{\Vec{z}}{S}{T}$;\\
  07. done;\\
  08. For each $B_i\in {\cal B}$ from the leaves to the root do \\
  09. \hspace*{2mm} Compute $\LPpdf{\tilde \phi}{D_i}{\Vec{w}}{S}{T}$ as the order $\tau$ 
  Taylor approximation of (\ref{form:tilde_phi_increment});\\
  10. done;\\
  11. Compute and output $\LPpdf{\tilde \Phi}{D_0}{x\Vec{\sigma}}{S}{T}$.
\end{algorithm}

\begin{proof} (of Theorem \ref{th:TaylorTreeWidth})
We first prove the theorem assuming
a separated tree decomposition with width $k$.
In the proof, we assume that all forms are expanded into the sum of products.
We store a form on memory as an array of coefficients of terms.

Consider the array size for 
$\LPpdf{\tilde \Phi}{G_i}{\Vec{z}}{S}{T}$
and $\LPpdf{\tilde \Phi}{D_i}{\Vec{z}}{S}{T}$.
Since we argue the integrations of step function factors 
in the same way as in the proof of
Theorem \ref{th:exponential_distribution},
the number of possible combinations of step function factors is
at most $(k+1)!$.
Since the polynomial factors of $\LPpdf{\tilde \phi}{G_i}{\Vec{z}}{S}{T}$ and
$\LPpdf{\tilde \phi}{D_j}{\Vec{z}}{S}{T}$
are degree $\tau$ polynomials of at most $k+1$ variables,
there are at most $\tau^{k+1}$
polynomial factors.
Thus, an array of $(k+1)!\tau^{k+1}$ 
real values
is sufficient to store $\LPpdf{\tilde \phi}{G_i}{\Vec{z}}{S}{T}$ and $\LPpdf{\tilde \phi}{D_i}{\Vec{z}}{S}{T}$.

Consider the running time for computing
$\LPpdf{\tilde \phi}{G_i}{\Vec{z}}{S}{T}$ in Step 02-04.
Remember the integrand of the form in
Theorem \ref{th:multiplesourceterminal}.
That is,
\begin{align*}
  \prod_{u\in \Src{G_i}\cup\internal{G_i}} \LPpdf{f}{u}{\Vec{z}}{u}{\Suc(u)} =\hspace*{-5mm}\prod_{u\in \Src{G_i}\cup \internal{G_i}}\sum_{v\in\Suc(u)} \frac{f_{uv}(z_u-z_v)}{F_{uv}(z_u-z_v)}\prod_{w\in \Suc(u)}F_{uw}(z_u-z_w).
\end{align*}
By expanding this form into the sum of products, we obtain
at most $k^k$ terms.
Each term is a product of edge length distribution functions or
density functions, where the number of factors is at most $k(k-1)$.
To compute the Taylor approximation of order $\tau$, 
we claim that we have at most $(\tau+1)^{k(k-1)}$ terms
in computing the order $\tau$ 
derivative of each term.
To bound the number of terms, let ${\cal F}_{uv}(z_u-z_v)$ be
either one of $f_{uv}(z_u-z_v)$, $F_{uv}(z_u-z_v)$ or $1$,
so that all terms are in the form of
$\prod_{uv\in \edge{G_i}}{\cal F}_{uv}(z_u-z_v)$.
For each $\Vec{d}=\{0,\dots,\tau\}^{\edge{G_i}}$,
the order $\tau$ derivative of each one of these terms is 
\begin{align*}
  \sum_{\Vec{d}\in \{0,\dots,\tau\}^{\edge{G_i}}}\frac{\tau!}{\prod_{uv\in \edge{G_i}}\Vec{d}[uv]!}\prod_{uv\in \edge{G_i}}{\cal F}_{uv}^{(\Vec{d}[uv])}(z_u-z_v),
\end{align*}
by the multinomial theorem,
where $\sum_{uv\in \edge{G_i}}\Vec{d}[uv]=\tau$.
Therefore, the number of the terms is at most
$\tau^{|\edge{G_i}|}\le \tau^{k(k-1)}$.
Since each term can be computed in $O(k(k-1))$ time,
the running time to obtain the order $\tau$ 
Taylor approximation is
$O(k(k-1)(\tau+1)^{k(k-1)})\le O(k^2 \tau^{k^2})$.
Thus, the running time to compute
$\LPpdf{\tilde \phi}{G_i}{\Vec{z}}{S}{T}$ is $O(k^{k+2} \tau^{k^2})$.

The most time-consuming part of
Step 08-10 is 
the multiplication of $\LPpdf{\tilde \phi}{G_i}{\Vec{z}}{S}{T}$ and
$\LPpdf{\tilde \phi}{U_i}{\Vec{z}}{S}{T}$.
We need $((k+1)!\tau^{k+1})^3$
arithmetic operations
to compute the coefficients of all
terms of the integrand of
$\LPpdf{\tilde \phi}{D_i}{\Vec{z}}{S}{T}$.
Since the integral for each polynomial term finishes in
$O(1)$ time, it takes 
$O((k+1)^{3k} \tau^{3k+3})$ time
to obtain $\LPpdf{\tilde \phi}{D_i}{\Vec{z}}{S}{T}$.
Therefore, we compute
$\LPpdf{\tilde \Phi}{D_0}{\Vec{z}}{S}{T}$ in $b k^{O(k)} \tau^{O(k)}$
time since
we repeat the above procedure for each vertex of each bag,
where $b$ is the number of bags and $b=O(n)$ by our assumption.
By the above arguments, we have the running time $O(nk^{O(k)}\tau^{O(k^2)})$
of our algorithm.

Next, we prove that
$\tau \ge (e^2+1)(k+1)x+ 2\ln b +\ln (1/\epsilon')+\ln 4-1$
is sufficient to bound the additive error of
$\LPpdf{\tilde \Phi}{D_0}{\Vec{z}}{S}{T}$ at most $\epsilon'$.
Let $\epsilon'_i(\Vec{z})$ satisfying $\LPpdf{\tilde \Phi}{D_i}{\Vec{z}}{S}{T}=\LPpdf{\Phi}{D_i}{\Vec{z}}{S}{T}+\epsilon'_i(\Vec{z})$
be the additive error of $\LPpdf{\tilde \Phi}{D_i}{\Vec{z}}{S}{T}$.
Let $\epsilon'_i$ be the maximum of $|\epsilon'_i(\Vec{z})|$ for
$\Vec{z}\in [0,x]^{S_i\cup T_i}$.
We have 
\begin{align*}
  \epsilon'_i&\le \underbrace{\int_{(0,x\Vec{1}[\Src{D_i}])}
    \LPpdf{\psi}{D_i}{\Vec{z}}{S}{T}{\rm d}\Vec{z}[\Src{D_i}]~-\LPpdf{\Phi}{D_i}{\Vec{z}}{S}{T}}_{\rm (A)}~+~
  \underbrace{\frac{((k+1)x)^{\tau+1}}{(\tau+1)!}}_{\rm (B)},
\end{align*}
where $\LPpdf{\psi}{D_i}{\Vec{z}}{S}{T}$ is defined
by (\ref{form:tilde_phi_increment}).
The earlier part (A) is the additive error produced in
the approximation before Step 09.
By Proposition \ref{proposition:Taylor_error},
the latter part (B) bounds the additive error
produced in computing the order $\tau$ 
Taylor approximation in Step 09.

Let $\LPpdf{\tilde\xi}{G_i}{\Vec{z}}{S}{T}\defeq(-1)^{|\JT_i|}\vecdiff{J_i}\LPpdf{\tilde\Phi}{G_i}{\Vec{z}}{S}{T}$.
In case $\Ch(i)=\{\ell,r\}$,
 we have $\text{(A)}=\int_{\mathbb{R}^{J_i}}\LPpdf{\tilde\xi}{G_i}{\Vec{z}}{S}{T}\LPpdf{\tilde \Phi}{U_i}{\Vec{z}}{S}{T}{\rm d}\Vec{z}[J_i]-\LPpdf{\Phi}{D_i}{\Vec{z}}{S}{T}$
by executing the integral of (\ref{form:tilde_phi_increment})
with respect to $s\in \Src{D_i}$ and
by Proposition~\ref{proposition:switch}.
Let $\gamma_i$ be the additive error of $\LPpdf{\tilde \xi}{G_i}{\Vec{z}}{S}{T}$
so that $\gamma_i=\max_{\Vec{z}\in [0,x]^{S_i\cup T_i}}|\LPpdf{\tilde \xi}{G_i}{\Vec{z}}{S}{T}-\LPpdf{\xi}{G_i}{\Vec{z}}{S}{T}|$, where
$\LPpdf{\xi}{G_i}{\Vec{z}}{S}{T}=(-1)^{|\JT_i|}\vecdiff{J_i}\LPpdf{\Phi}{G_i}{\Vec{z}}{S}{T}$.
We assume that $\epsilon'_\ell \ge \epsilon'_r$
without loss of generality.
Then, we have 
\begin{align*}
  \text{(A)}\le &\int_{\mathbb{R}^{J_i}}(\LPpdf{\xi}{G_i}{\Vec{z}}{S}{T}+\gamma_i)\prod_{j\in\Ch(i)}(\LPpdf{\Phi}{D_j}{\Vec{z}}{S}{T}+\epsilon'_j){\rm d}\Vec{z}[J_i]-\LPpdf{\Phi}{D_i}{\Vec{z}}{S}{T}.
\end{align*}
By canceling $\LPpdf{\Phi}{D_i}{\Vec{z}}{S}{T}$ using Proposition~\ref{proposition:switch}, the right hand side is at most
\begin{align*}
  &\int_{\mathbb{R}^{J_i}}\left(\LPpdf{\Phi}{D_\ell}{\Vec{z}}{S}{T}\epsilon'_r\LPpdf{\xi}{G_i}{\Vec{z}}{S}{T}+\epsilon'_\ell\LPpdf{\Phi}{D_r}{\Vec{z}}{S}{T}\LPpdf{\xi}{G_i}{\Vec{z}}{S}{T} \right){\rm d}\Vec{z}[J_i]\\
  &+\int_{\mathbb{R}^{J_i}}\left(\epsilon'_\ell\epsilon'_r\LPpdf{\xi}{G_i}{\Vec{z}}{S}{T}+\gamma_i\prod_{j\in \Ch(i)}(\LPpdf{\Phi}{D_j}{\Vec{z}}{S}{T}+\epsilon'_j)\right){\rm d}\Vec{z}[J_i]\\
  &\le \int_{\mathbb{R}^{J_i}}\left( (\epsilon'_\ell+2\epsilon'_r)\LPpdf{\xi}{G_i}{\Vec{z}}{S}{T}+4\gamma_i\right){\rm d}\Vec{z}[J_i].
\end{align*}
The last inequality is because
$\LPpdf{\Phi}{D_\ell}{\Vec{z}}{S}{T},\LPpdf{\Phi}{D_r}{\Vec{z}}{S}{T}\le 1$,
and $\epsilon'_r\le \epsilon'_\ell\le 1$.
Executing the integrals with respect to $\Vec{z}[J_i]$,
we have $\int_{\mathbb{R}^{J_i}}\LPpdf{\xi}{G_i}{\Vec{z}}{S}{T}{\rm d}\Vec{z}[J_i]=\lim_{\Vec{z}[J_i]\rightarrow \infty}\LPpdf{\Phi}{G_i}{\Vec{z}}{S}{T}\le 1$.
Now, we have $(A)\le \epsilon'_\ell+2\epsilon'_r + 4 x^{k+1} \gamma_i$.

Consider the bound of $\gamma_i$.
By the assumption, the integrand (and its derivatives)
of Theorem \ref{th:multiplesourceterminal} can be computed exactly
for each bag-subgraph $G_i$.
Then, we compute the order $\tau$ 
Taylor approximation of the integrand
that produces the additive error $\gamma_i$ given in
Proposition~\ref{proposition:Taylor_error},
where the result is a (at most) $k+1$ variable function;
$R(k,\tau,x)\defeq \frac{((k+1)x)^{\tau+1}}{(\tau+1)!}$
is the upper bound on $\gamma_i$.
The above arguments imply that
$\epsilon'_i\le \epsilon'_\ell+2\epsilon'_r + (4x^{k+1}+1)R(k,\tau,x)$.

We consider the total additive error
$\epsilon'_0$. 
Since we set $\epsilon'_r=0$ in case $|\Ch(i)|=1$,
the worst is the case in which tree decomposition ${\cal T}(G)$
is a complete binary tree
where $\epsilon'_\ell=\epsilon'_r$.
We bound the additive error from the leaves to the root.
For leaf bag $B_c\in {\cal B}$,
we have $\epsilon'_c=\delta'_c\le R(k,\tau,x)$. 
As the computation advances toward the root,
the error at the children is multiplied by three 
and added by
$(4x^{k+1}+1)R(k,\tau,x)$.
Thus,
\begin{align*}
  \epsilon'_0\le \lceil \log_2 b \rceil 3^{\lceil \log_2 b\rceil}(4x^{k+1}+1)R(k,\tau,x)\le b^2(4x^{k+1}+1)R(k,\tau,x).
\end{align*}

Let us estimate sufficiently large $\tau$ 
to achieve
$\epsilon'_0\le \epsilon'$.
By considering the difference of the logarithms of
$\epsilon'$ and $\epsilon'_0$,
we have
\begin{align*}
  &\ln \epsilon'- \ln \epsilon'_0\ge-\ln(1/\epsilon')-\ln b^2(4x^{k+1}+1)- \ln \frac{((k+1)x)^{\tau+1}}{(\tau+1)!},
\end{align*}
which is at least
$-\ln(1/\epsilon')-2\ln b-\ln 4-(k+1)\ln(x+1)+(\tau+1)\left(\ln \frac{\tau+1}{e(k+1)x}\right)$
by Stirling's approximation
$n!\sim \sqrt{2\pi n}\left(\frac{n}{e}\right)^n\ge \left(\frac{n}{e}\right)^n$.
Now,
\begin{align*}
  \tau+1 \ge \max\{e^2(k+1)x, 2\ln b + (k+1)\ln(x+1)+\ln (1/\epsilon')+\ln 4\}
\end{align*}
is sufficient to make this lower bound positive.
By taking the sum instead of max, we have the theorem since $\ln 4=1.386...$.
We have the theorem by replacing $k$ with $2k+1$
to construct a separated tree decomposition.
We process the delta functions for edges in $E_V$
in the same way as in the proof of Theorem~\ref{th:exponential_distribution}.
\end{proof}

\section*{Acknowledgment}
This work was partly supported by JSPS KAKENHI Grant Number 19K11832.
The author is grateful to the anonymous reviewers for the valuable advice and comments.




\end{document}